\documentclass[11pt]{article}

\usepackage{tikz}

\usepackage{amsmath,amsfonts,amsthm,amssymb,color}

\newtheorem{theorem}{Theorem}[section]
\newtheorem{corollary}[theorem]{Corollary}
\newtheorem{lemma}[theorem]{Lemma}
\newtheorem{proposition}[theorem]{Proposition}
\newtheorem{definition}[theorem]{Definition}

\theoremstyle{remark}

\usepackage{latexsym,bbm,xspace,graphicx,float}
\definecolor{newblue}{rgb}{0.19, 0.55, 0.91}
\usepackage[colorlinks,citecolor=newblue,bookmarks=true,pagebackref=true, urlcolor=blue, linkcolor=blue, linktoc=page]{hyperref}
\usepackage[nameinlink]{cleveref}
\usepackage[letterpaper,margin=1in]{geometry}

\usepackage{url}
\usepackage{mathtools}
\usepackage{graphicx}
\usepackage{multirow}
\usepackage[ruled,linesnumbered]{algorithm2e}

\newcommand{\eps}{\varepsilon}

\usepackage{tikz}
\usepackage{mdframed}
\usepackage{color, colortbl}

\def\1{\bm{1}}

\DeclareMathAlphabet{\mathsfit}{\encodingdefault}{\sfdefault}{m}{sl}
\SetMathAlphabet{\mathsfit}{bold}{\encodingdefault}{\sfdefault}{bx}{n}

\DeclareMathOperator*{\E}{\mathbb{E}}

\newcommand{\R}{\mathbb{R}}

\DeclareMathOperator*{\argmin}{arg\,min}

\renewcommand{\tilde}{\widetilde}
\renewcommand{\bar}{\overline}

\newcommand{\Norm}[1]{\left\|#1\right\|}
\newcommand{\norm}[1]{\|#1\|}

\providecommand{\expect}[2]{\ensuremath{\ifthenelse{\equal{#1}{}}{\mathbb{E}}{\mathbb{E}_{#1}}\!\left[#2\right]}\xspace}
\providecommand{\prob}[2]{\ensuremath{\ifthenelse{\equal{#1}{}}{\Pr}{\Pr_{#1}}\!\left[#2\right]}\xspace}

\newcommand{\nnz}{\mathrm{nnz}\xspace}
\newcommand{\inner}[1]{\langle #1\rangle}

\usepackage{bm}

\newcommand{\abs}[1]{\left|{#1}\right|}
\DeclareMathOperator{\diag}{diag}
\DeclareMathOperator{\poly}{poly}

\newcommand{\bado}{\mathsf{Bad}}
\newcommand{\rank}{\mathrm{rank}}
\newcommand{\spa}{\mathrm{span}}

\usepackage[inline]{enumitem}
\usepackage{array}
\newcolumntype{L}[1]{>{\raggedright\let\newline\\\arraybackslash\hspace{0pt}}m{#1}}
\newcolumntype{C}[1]{>{\centering\let\newline\\\arraybackslash\hspace{0pt}}m{#1}}
\newcolumntype{R}[1]{>{\raggedleft\let\newline\\\arraybackslash\hspace{0pt}}m{#1}}

\title{}

\begin{document}

\author{Yi Li\\ \small{Division of Mathematical Sciences}\\ \small{Nanyang Technological University}\\  \small{\texttt{yili@ntu.edu.sg}}\\
	   \and
	   Honghao Lin$\qquad\qquad$ David P. Woodruff \\ \small{Computer Science Department}\\ \small{Carnegie Mellon University} \\  \small{\texttt{\{honghaol,dwoodruf\}@andrew.cmu.edu}}
	   %
	   }
    

\date{\vspace{-5ex}}	   

\title{$\ell_p$-Regression in the Arbitrary Partition Model of Communication}

\maketitle

\newcolumntype{L}[1]{>{\raggedright\let\newline\\\arraybackslash\hspace{0pt}}m{#1}}
\newcolumntype{C}[1]{>{\centering\let\newline\\\arraybackslash\hspace{0pt}}m{#1}}
\newcolumntype{R}[1]{>{\raggedleft\let\newline\\\arraybackslash\hspace{0pt}}m{#1}}



\maketitle

\begin{abstract}
   We consider the randomized communication complexity of the distributed $\ell_p$-regression problem in the coordinator model, for $p\in (0,2]$. In this problem, there is a coordinator and $s$ servers. The $i$-th server receives $A^i\in\{-M, -M+1, \ldots, M\}^{n\times d}$ and $b^i\in\{-M, -M+1, \ldots, M\}^n$ and the coordinator would like to find a $(1+\eps)$-approximate solution to $\min_{x\in\R^n} \norm{(\sum_i A^i)x - (\sum_i b^i)}_p$. Here $M \leq \poly(nd)$ for convenience. This model, where the data is additively shared across servers, is commonly referred to as the arbitrary partition model. 
   
   We obtain significantly improved bounds for this problem. For $p = 2$, i.e., least squares regression, we give the first optimal bound of $\tilde{\Theta}(sd^2 + sd/\epsilon)$ bits. 
   
   For $p \in (1,2)$,
   we obtain an 
   $\tilde{O}(sd^2/\eps + sd/\poly(\eps))$ upper bound. Notably, for $d$ sufficiently large, our leading order term only depends linearly on $1/\epsilon$ rather than quadratically. 
   We also show communication lower bounds of $\Omega(sd^2 + sd/\eps^2)$ for $p\in (0,1]$ and $\Omega(sd^2 + sd/\eps)$ for $p\in (1,2]$. Our bounds considerably improve previous bounds due to (Woodruff et al. COLT, 2013) and (Vempala et al., SODA, 2020). 
\end{abstract}

\section{Introduction}
Regression is a lightweight machine learning model used to capture linear dependencies between variables in the presence of noise. In this problem there is a (sometimes implicit) matrix $A \in \mathbb{R}^{n \times d}$ and a vector $b \in \mathbb{R}^n$ and the goal is to find a hyperplane $x \in \mathbb{R}^d$ for which $\|Ax-b\|$ is small for some loss function $\|\cdot\|$, which throughout this paper will be a norm. Here $A$ is known as the design matrix, $b$ the response vector, and $x$ the model parameters. We focus on the over-constrained case, when $n \gg d$, which corresponds to having many more examples than features. Although more sophisticated models can often achieve lower error, regression is often the most computationally efficient and the first model of choice.   

One of the most popular loss functions is the $\ell_p$-norm, or equivalently its $p$-th power $\|y\|_p^p = \sum_{i=1}^n |y_i|^p$. When $p = 2$ this is least squares regression, which corresponds to the maximum likelihood estimator (MLE) in the presence of Gaussian noise. When the noise is more heavy-tailed, often $p < 2$ is chosen as the loss function since it is more robust to outliers. Indeed, since one is not squaring the differences, the optimal solution pays less attention to large errors. For example, $p = 1$ gives the MLE for Laplacian noise. While $p < 1$ results in non-convex loss functions, heuristics are still used given its robustness properties. When $p > 2$, the loss function is even more sensitive to outliers; it turns out that such $p$ cannot be solved without incurring a polynomial dependence on $n$ in the communication model we study, see below, and so our focus will be on $p \leq 2$. 

It is often the case that data is either collected or distributed across multiple servers and then a key bottleneck is the {\it communication complexity}, i.e., the number of bits transmitted between the servers for solving a problem. We consider the standard coordinator model of communication, also known as the message-passing model, in which there is a site designated as the coordinator who has no input, together with $s$ additional sites, each receiving an input. There is a communication channel between the coordinator and each other server, and all communication goes through the coordinator. This model is convenient since it captures arbitrary point-to-point communication up to small factors, i.e., if server $i$ wants to send a message to server $j$, server $i$ can first send the message to the coordinator and then have it forwarded to server $j$. We note that in addition to the total communication, it is often desirable to minimize the time complexity on each server, and the protocols in this paper will all be time-efficient. 

A natural question in any communication model is how the input is distributed. We study the {\it arbitrary partition model} of \cite{KVW14,BWZ16}, which was studied for the related task of low rank approximation. In this model, the $i$-th server receives $A^i\in\{-M, -M+1, \ldots, M\}^{n\times d}$ and $b^i\in\{-M, -M+1, \ldots, M\}^n$ and the coordinator would like to find a $(1+\eps)$-approximate solution to $\min_{x\in\R^n} \norm{(\sum_i A^i)x - (\sum_i b^i)}_p$. Here $M \leq \poly(nd)$ for convenience. Note that this model gives more flexibility than the so-called {\it row partition model} in which each example and corresponding response variable is held on exactly one server, and which is a special case of the arbitrary partition model. For example, if each row $i$ of $A$ corresponds to an item and each column $j$ to a user and an entry $A_{i,j}$ corresponds to the number of times user $i$ purchased item $j$, then it might be that each server $t$ is a different shop where the user could purchase the item, giving a value $A^t_{i,j}$, and we are interested in $\sum_{t= 1}^s A^t_{i,j}$, i.e., the matrix which aggregates the purchases across the shops. This communication model is also important for {\it turnstile streaming} where arbitrary additive updates are allowed to an underlying vector \cite{M05}, as low-communication protocols often translate to low memory streaming algorithms, while communication lower bounds often give memory lower bounds in the streaming model. The number of communication rounds often translates to the number of passes in a streaming algorithm. See, e.g., \cite{BWZ16}, as an example of this connection for low rank approximation. We note that for $p > 2$, there is an $\Omega(n^{1-2/p})$ lower bound in the arbitrary partition model even for just estimating the norm of a vector \cite{BJKS04,G09,J09}, and so we focus on the $p < 2$ setting. 

The communication complexity of approximate regression was first studied in the coordinator model in the row partition model in \cite{wz13}, though their protocols for $1 \leq p < 2$  use $\tilde{O}(sd^{2+\gamma} + d^5 + d^{3+p}/\eps^2)$ communication, where $\tilde{O}(f)$ suppresses a $\poly(\log(sdn/\eps))$ factor. These bounds were later improved in the coordinator model and in the row partition model in \cite{VWW20}, though the bounds are still not optimal, i.e., their lower bounds do not depend on $\eps$,  are suboptimal in terms of $s$, or hold only for deterministic algorithms. Their upper bounds also crucially exploit the row partition model, and it is unclear how to extend them to the arbitrary partition model. We will substantially improve upon these bounds. 

Despite the previous work on understanding the communication complexity of a number of machine learning models (see, e.g., \cite{VWW20} and the references therein), perhaps surprisingly for arguably the most basic task of regression, the optimal amount of communication required was previously unknown. 

\paragraph{Our Results} We obtain a lower bound of $\Omega(sd^2+sd/\eps^2)$ for $p\in(0,1]$ and a lower bound of $\Omega(sd^2 + sd/\eps)$ for $p\in (1,2]$, both of which improve the only known lower bound of $\tilde{\Omega}(d^2 + sd)$ by \cite{VWW20}. We strengthen their $d^2$ lower bound by a multiplicative factor of $s$ and incorporate the dependence on $\eps$ into their $sd$ lower bound. 

When $p=2$, we obtain an upper bound of $\tilde{O}(sd^2 + sd/\eps)$ bits, which matches our lower bound up to logarithmic factors. The total runtime of the protocol is $O(\sum_i \nnz(A^i) + s\poly(d/\eps))$, which is optimal in terms of $\nnz(A^i)$. Here for a matrix $A$, $\nnz(A)$ denotes the number of non-zero entries of $A$. Our results thus largely settle the problem in the case of $p = 2$. 

When $p\in (1,2)$, we obtain an upper bound of $\tilde{O}(sd^2/\eps + sd/\poly(\eps))$ bits with a runtime of $O(\sum_i \nnz(A^i) (d/\eps^{O(1)}) + s \poly(d/\eps))$. Note that if the $\tilde{O}(sd^2/\eps)$ term dominates, then our upper bound is optimal up to a $1/\eps$ factor due to our lower bound. Interestingly, this beats a folklore sketching algorithm for which each server sketches their input using a shared matrix of $p$-stable random variables with $\tilde{O}(d/\eps^2)$ rows, sends their sketch to the coordinator with $\tilde{O}(sd^2/\eps^2)$ total communication, and has the coordinator add up the sketches and enumerate over all $x$ to find the best solution (see, e.g., Appendix F.1 of \cite{BIPRW15} for a proof of this for $p = 1$). Moreover, our algorithm is time-efficient, while the sketching algorithm is not. In fact, any sketch that solves the harder problem of computing an $\ell_p$-subspace embedding requires $\poly(d)$ distortion \cite{WW19} or has an exponential dependence on $1/\eps$ \cite{LWY21}. We further show that if the leverage scores of $[A\ b]$ are uniformly small, namely, at most $\poly(\eps)/d^{4/p}$, then our runtime can be improved to $O(\sum_i \nnz(A^i) + s \poly(d/\eps))$, which is now optimal in terms of $\nnz(A)$, with the same amount of communication. Along the way we prove a result on embedding  $d$-dimensional subspaces in $\ell_p^n$ to $\ell_r$ for $1 < r < p$, which may be of independent interest.

\begin{table}[t]
\centering
\begin{tabular}{ c c c c }
&  & Communication \\
\hline
$0 < p < 2$ & Upper Bound & $\tilde{O}(sd^2/\eps^2)$ & Folklore \\
$p = 2$ & Upper Bound & $\tilde{O}(sd^2/\eps)$ & \cite{CW09} \\
$0 < p \leq 2$ & Lower Bound  & $\Omega(d^2 + sd)$ & \cite{VWW20} \\
$p = 1$ & Upper Bound$^*$  & $\tilde{O}(\min(sd^{2} + \frac{d^2}{\eps^2}, \frac{sd^3}{\eps}))$ &  \cite{VWW20}\\
$p = 2$ & Upper Bound$^*$  & $\tilde{O}(sd^{2})$ &  \cite{VWW20}\\
$1 \le p < 2$ & Upper Bound$^*$  & $\tilde{O}(sd^{2 + \gamma} + d^5 + d^{3 + p}/\eps^2)$ &  \cite{wz13}\\
\rowcolor{blue!15} $0 < p \le 1$ & Lower Bound & $\Omega(sd^2 + sd/\eps^2)$& Theorem~\ref{thm:lb_d/eps}, \ref{thm:lb_sd^2}\\
\rowcolor{blue!15}$1 < p \le 2$ & Lower Bound & $\Omega(sd^2 + sd/\eps)$& Theorem~\ref{thm:lb_d/eps}, \ref{thm:lb_sd^2}\\
\rowcolor{blue!15}$1 < p < 2$ & Upper Bound & $\tilde{O}(sd^2/\eps + sd/\poly(\eps))$& Theorem~\ref{thm:ell_p_regression}\\
\rowcolor{blue!15}$p = 2$ & Upper Bound & $\tilde{O}(sd^2 + sd/\eps)$& Theorem~\ref{thm:ell_2_regression}\\
\end{tabular}
\caption{Summary of the results for the distributed $\ell_p$ regression problem. $^*$ denotes row partition model. The upper bound in the first row uses a median sketch of the $p$-stable distribution, which is time-inefficient, see, e.g., Section F.1 of~\cite{BIPRW15}. 
}
\label{tab:our_result}
\end{table}

\paragraph{Open Problems} 
We leave several intriguing questions for future work. 

First, it would be good to close the gap in our upper and lower bounds as a function of $\eps$ for $p < 2$. For $1 < p < 2$, if $\poly(1/\eps) < d$ then our bounds are off by a $1/\eps$ factor, namely, our upper bound is $\tilde{O}(sd^2/\eps)$, but our lower bound is $\Omega(sd^2)$. 

Second, the $\nnz$ term in our runtime in general has a multiplicative factor of $d/\poly(\eps)$. This is mainly due to the use of a dense matrix for the lopsided subspace embedding of $\ell_p^n$ into $\ell_r$, and it is interesting to see whether there are sparse lopsided subspace embeddings of $\ell_p^n$ into $\ell_r$. 

\subsection{Our Techniques}

\paragraph{Lower Bounds}

We first demonstrate how to show an $\Omega(sd/\eps^2)$ lower bound for $p\in (0,1]$ and an $\Omega(sd/\eps)$ lower bound for $p\in (1,2]$. 

Let us first consider the special case of $d = 1$.
Consider the $\ell_p$ regression problem $\min_{x \in \R} \norm{a \cdot x - b}_p$, where $a$ and $b$ are uniformly drawn from $\{-1, 1\}^n$. The crucial observation is that the solution $x$ reveals the Hamming distance $\Delta(a,b)$. Specifically, when $n = \Theta(1/\eps^2)$, a $(1 \pm \eps)$-solution when $0 < p \le 1$ and $(1 \pm \eps^2)$-solution when $1 < p \le 2$ suffice for us to solve the Gap-Hamming communication problem (\textsf{GHD}) of $a$ and $b$ (determining $\Delta(a, b) \ge c \sqrt{n}$ or $\Delta(a, b) \le -c \sqrt{n}$). The \textsf{GHD} problem has an $\Omega(n)$ information cost lower bound~\cite{BGPW16}, which implies, by our choice of $n$, an $\Omega(1/\eps^2)$ lower bound for $p\in (0,1]$ and an $\Omega(1/\eps)$ lower bound for $p\in (1,2]$.

To gain the factor of $s$, we design a distributed version of \textsf{GHD}, the $s$-GAP problem, as follows. There are $2s$ players. Each of the first $s$ players holds a vector $a^i \in \{-1, 1\}^n$ and each of the remaining players holds a $b^i \in \{-1, 1\}^n$, with the guarantee that $\sum_i a^i = a$ and $\sum_i b^i = b$. The $2s$ players and the coordinator will collectively determine the two cases of $\Delta(a, b)$. Our goal is to show an $\Omega(sn)$ lower bound for this communication problem. To this end, we employ the symmetrization technique that was used in~\cite{PVZ16}. Specifically, Alice simulates a random player and Bob the remaining $s - 1$ players. As such, Bob will immediately know the whole vector $b$ and part of the vector $a$ (denote the set of these indices by $I$). As we will show in the proof, to determine the distance $\Delta(a, b)$, Alice and Bob still need to approximately determine $\Delta(a_{I^c}, b_{I^c})$, which requires $\Omega(|I^c|) = \Omega(n)$ communication. Note that the input distribution of each player is the same and Alice is choosing a random player. Hence, Alice's expected communication to Bob is at most $O(\chi/s)$ bits if $s$-GAP can be solved using $\chi$ bits of communication, which yields a lower bound of $\Omega(sn)$ bits for the $s$-GAP problem. 

So far we have finished the proof for $d=1$. To obtain a lower bound for general $d$, we use a padding trick. Consider $A = \diag(a_1,\dots,a_d)$ and let $b$ be the vertical concatenation of $b_1,\dots,b_d$, where each pair $(a_i,b_i)$ is drawn independently from the hard distribution for $d=1$. One can immediately observe that $\min_x \norm{Ax-b}_p^p = \sum_i \min_{x_i}\norm{a_ix_i-b}_p^p$ and show that approximately solving $\min_x \norm{Ax-b}_p^p$ can approximately solve a constant fraction of the $d$ subproblems $\min_{x_i}\norm{a_ix_i-b}_p^p$. This further adds an $O(d)$ factor to the lower bound.

Next we discuss the $\Omega(sd^2)$ lower bound. We shall follow the idea of~\cite{VWW20} and construct a set of matrices $\mathcal{H} \subseteq \{-1, 1\}^{d \times d}$ with a vector $b \in \mathbb{R}^d$ such that (i) $A$ is non-singular for all $A \in \mathcal{H}$, (ii) $A^{-1} b \ne B^{-1}b$ for all $A, B \in \mathcal{H}$ and $A\neq B$ and (iii) $\abs{\mathcal{H}} = 2^{\Omega(d^2)}$. The conditions (i) and (ii) mean that a constant-factor approximation to $\min_x \norm{Ax - b}_p^p$ is exact, from which the index of $A$ in the set $\mathcal{H}$ can be inferred. Condition (iii) then implies an $\Omega(d^2)$ lower bound for solving the regression problem up to a constant factor. To gain a factor of $s$, we consider the communication game where the $i$-th player receives a matrix $A^i\subseteq \{-1,1\}^{d\times d}$ with the guarantee that $A = \sum_{i} A^i$ is distributed in $\mathcal{H}$ uniformly. Then the $s$ players with the coordinator want to recover the index of $A$ in $\mathcal{H}$. We consider a similar symmetrization technique. However, the issue here is if Bob simulates $s - 1$ players, he will immediately know roughly a $\frac{1}{2}$ fraction of coordinates of $A$, which can help him to get the index of $A$ in $\mathcal{H}$. To overcome this, we choose a different strategy where Alice simulates two (randomly chosen) players and Bob simulates the remaining $s - 2$ players. In this case Bob can only know a $\frac{1}{4}$-fraction of the coordinates without communication. However, one new issue here is Bob will know partial information about the remaining coordinates. But, as we shall show in the proof, even when conditioned on Bob's input on $s - 2$ players, with high probability the entropy of the remaining coordinates is still $\Omega(d^2)$. This implies that Alice still needs to send $\Omega(d^2)$ bits to Bob, which yields an $\Omega(sd^2)$ lower bound for the original problem.

\paragraph{Upper Bounds}  

For the $\ell_p$-regression $\min_x\norm{Ax-b}_p$, a classical ``sketch-and-solve'' approach is to use a $(1+\eps)$-subspace embedding $S$ for $B = [A\ b]\in \R^{n\times(d+1)}$ and reduce the problem to solving $\min_x \norm{SAx-Sb}_p$, which is of much smaller size. The subspace embedding is non-oblivious and obtained by subsampling $\tilde{O}(d/\eps^2)$ rows of $B$ with respect to the Lewis weights of $B$~\cite{CP15}. More recently, it was shown that sampling $\tilde{O}(d/\eps)$ rows according to the Lewis weights is sufficient for solving  $\ell_p$-regression~\cite{MMW+22,CLS22}, instead of $\tilde{O}(d/\eps^2)$ rows needed for an  $\ell_p$-subspace embedding. However, computing the Lewis weights is expensive and would incur a communication cost as well as a runtime at least linear in $n$, which is prohibitive in our setting. 

Instead of embedding an $\ell_p$-subspace into $\ell_p$, we $(1+\eps)$-embed an $\ell_p$-subspace into $\ell_r$ for some $1 < r < p$. Furthermore, since we are solving a regression problem, we do not need a conventional subspace embedding but only a \emph{lopsided} one; that is, the map $S$ must not contract $\norm{Ax-b}_p$ for all $x$ simultaneously but it is required not to dilate $\norm{Ax^\ast-b}_p$ for only the optimal solution $x^\ast$. We show that an $S$ of i.i.d.\ $p$-stable variables and $O(d\log d/\poly(\eps))$ rows suffices (see Lemma~\ref{lem:p_stable_embedding} for the formal statement). Such a lopsided subspace embedding for embedding a subspace of $\ell_p^n$ into $\ell_r$, to the best of our knowledge, has not appeared in the  literature\footnote{We note that the works of~\cite{Pisier83,FG11} consider embedding the entire space $\ell_p^n$ into $\ell_r$ instead of embedding a low-dimensional subspace of $\ell_p^n$ into $\ell_r$.} and may be of independent interest. This lopsided subspace embedding reduces the $\ell_p$ regression problem to an $\ell_r$-regression problem of $\tilde{O}(d/\poly(\eps))$ rows. Importantly though, we do not need to ever explicitly communicate these rows in their entirety. Namely, we can leave the regression problem in an implicit form and now run a Lewis weight approximation algorithm, and since our effective $n$ has been replaced with $d/\poly(\eps)$, we just need $d/\poly(\eps)$ communication to iteratively update each of the weights in the Lewis weight algorithm, rather than $n$ communication. 

For the $\ell_2$-regression problem, it is known that a $(1+\sqrt{\eps})$-subspace embedding can yield a $(1+\eps)$-approximate solution (see, \cite{BN13}, also the [Woo14] reference therein) and so the subspace embedding $S$ needs only to have $O(d (\log d) /\eps)$ rows. The servers then run gradient descent on the sketched version $\min_x\norm{SAx-Sb}_2$. To ensure fast convergence in $O(\log(1/\eps))$ iterations, the servers will instead solve $\min_x\norm{SARx-Sb}_2$, where $R$ is a pre-conditioner to make $SAR$ have a constant condition number. Putting these pieces together leads to our near-optimal communication and runtime.
\section{Preliminaries}
\label{sec:prelim}

\paragraph{$\ell_2$ Subspace Embeddings.} 
For a matrix $A\in \R^{n\times d}$, we say a matrix $S\in \R^{m\times n}$ is a $(1\pm\eps)$-$\ell_2$ subspace embedding for the column span of $A$ if $(1-\eps)\norm{Ax}_2\leq \norm{SAx}_2 \leq (1+\eps)\norm{Ax}_2$ for all $x\in \R^d$ with probability at least $1 - \delta$. 
We summarize the subspace embeddings we use in this paper below:

\begin{itemize}[nosep]
	\item \textsf{Count-Sketch}: $m = O(d^2/(\delta \eps^2))$ with $s = 1$ non-zero entry per column, with each non-zero entry in $\{-1, 1\}$~\cite{CW17}. Computing $SA$ takes only $O(\nnz(A))$ time.
	\item \textsf{OSNAP}: $m=O((d \log (d / \delta)) /\eps^2)$ and has $s=O((\log (d / \delta))/\eps)$ non-zeros per column, with each non-zero entry in $\{-1, 1\}$ ~\cite{NN13,Cohen16}. Computing $SA$ takes $O(s \cdot  \nnz(A)) = O(\nnz(A)(\log (d/\delta)/\eps))$ time.
\end{itemize}

\paragraph{$p$-stable Distributions.}

Our protocol for distributed $\ell_p$ regression will use $p$-stable distributions, which are defined below.

\begin{definition}[\cite{zolotarev1986one}]
For $0 < p < 2$, there exists a probability distribution $\mathcal{D}_p$ called the $p$-stable distribution, which satisfies the following property. For any positive integer $n$ and vector $x \in \R_n$, if $Z_1, \ldots, ,Z_n \sim \mathcal{D}_p$ are independent, then $\sum_{j=1}^n Z_jx_j \sim  \|x\|_pZ$ for $Z \sim \mathcal{D}_p$.
\end{definition}

\paragraph{Lewis Weights.} Below we recall some facts about Lewis weights. For more details, we refer the readers to, e.g., \cite[Section 3.3]{CWW19}.

\begin{definition}\label{def:leverage_score}
Given a matrix $A \in \mathbb{R}^{n \times d}$. The {\em leverage score} of a row $A_{i, *}$ is defined to be
$
\tau_i(A) =  A_{i, *} (A^TA)^{\dagger}  (A_{i, *})^T
$.
\end{definition}

\begin{definition}[\cite{CP15}]\label{def:lewis}
For a matrix $A \in \mathbb{R}^{n \times d}$, its  $\ell_p$-Lewis weights
$\{w_i\}_{i=1}^n$ are the unique weights such that 
$
w_i = \tau_i(W^{1/2 - 1/p} A)
$ for each $i \in [n]$.
Here $\tau_i$ is the leverage score of the $i$-th row of a matrix
and
$W$ is the diagonal matrix whose diagonal entries are $w_1,\dots,w_n$.
\end{definition}

The Lewis weights are used in the construction of randomized $\ell_p$-subspace embeddings. In particular, the rescaled sampling matrix w.r.t.\ Lewis weights gives an $\ell_p$-subspace embedding.
\begin{definition}
Given $p_1,\dots,p_n\in [0,1]$ and $p\geq 1$, the \emph{rescaled sampling matrix} $S$ with respect to $p_1,\dots,p_n$ is a random matrix formed by deleting all zero rows from a random $n\times n$ diagonal matrix $D$ in which $D_{i,i} = p_i^{-1/p}$ with probability $p_i$ and $D_{i,i} = 0$ with probability $1-p_i$.
\end{definition}

\begin{lemma}[Lewis weight sampling, \cite{CP15}]\label{lem:lewis_weights_sampling}
Let $A\in \R^{n\times d}$ and $p\geq 1$. Choose an oversampling parameter $\beta = \Theta(\log(d/\delta)/\eps^2)$ and sampling probabilities $p_1,\dots, p_n$ such that $\min\{\beta w_i(A),1\}\leq p_i\leq 1$ and let $S$ be the rescaled sampling matrix with respect to $p_1,\dots,p_n$. Then it holds with probability at least $1-\delta$ that $(1-\eps)\norm{Ax}_p\leq \norm{SAx}_p \leq (1+\eps)\norm{Ax}_p$ (i.e., $S$ is an  $\eps$-subspace embedding for $A$ in the $\ell_p$-norm) and $S$ has $O(\beta \sum_i w_i(A)) = O(\beta d)$ rows.
\end{lemma}

\cite{CP15} give an iterative algorithm (Algorithm~\ref{alg:lewis}) which computes the Lewis weights time-efficiently for $p < 4$. 
\begin{lemma} [\cite{CP15}]
\label{lem:lewis_algortihm}
    Suppose that $p<4$ and $\beta = \Theta(1)$. After $T = \log \log (n)$ iterations in Algorithm~\ref{alg:lewis}, $w$ is a constant approximation to the $\ell_p$ Lewis weights.  
\end{lemma}



\begin{algorithm}[tb]
\begin{mdframed}
\begin{enumerate}[nosep]
    \item Initialize $w = \mathbf{1} \in \R^n$.
    \item For $t = 1, 2, \dots, T$
        \begin{enumerate}[nosep]
            \item Let $\tau \in \R^n$ be a $\beta$-approximation of the leverage scores of $W^{1/2 - 1/p} A$.
            \item Set $w_i \leftarrow (w_i^{2/p -1} \tau_i)^{p/2}$.
        \end{enumerate}
    \item Return $w$.  
\end{enumerate}
\end{mdframed}
\caption{Iterative Algorithm to Compute the $\ell_p$ Lewis Weights}
\label{alg:lewis}
\end{algorithm}

\section{Distributed $\ell_p$-Regression Lower Bound} 

We consider the following variant of the Gap-Hamming problem (\textsf{GHD}).

\paragraph{Gap-Hamming Problem.} In the Gap-Hamming problem ($\textsf{GHD}_{n, c}$), Alice and Bob receive binary strings $x$ and $y$, respectively, which are uniformly sampled from $\{-1,1\}^n$. 
They wish to decide which of the following two cases $\Delta(x, y) = \sum_{i=1}^n x_i y_i$ falls in: $\Delta(x, y) \ge c\sqrt{n}$ or $\Delta(x, y) \le  -c\sqrt{n}$, where $c$ is a constant. (If $\Delta(x, y)$ is between $- c\sqrt{n}$ and $c\sqrt{n}$, an arbitrary output is allowed.) 

\begin{lemma}[\cite{BGPW16}]
\label{lem:gap}
If there is a protocol $\Pi$ which solves  $\textsf{GHD}_{n, c}$ with large constant probability, then we have $I(x, y ; \Pi) = \Omega(n)$, where $I$ denotes mutual information and the constant hidden in the $\Omega$-notation depends on $c$.
\end{lemma}

\subsection{$s$-GAP problem}

In this section, we will define the $s$-GAP problem and then prove an $\Omega(sn)$ lower bound.

\begin{definition}
\label{def:k_gap}
In the $s$-GAP problem, there are $2s$ players, where for the first $s$ players, the $i$-th player receives an $n$-bit string $a^i \in \{-1, 1\}^n$, and for the remaining $s$ players, the $i$-th player receives an $n$-bits string $ b^i \in \{-1, 1\}^n$, with the guarantee that $a = \sum_i a^i \in \{-1,1\}^n$, $b = \sum_i b^i \in \{-1, 1\}^n$ and $\Delta(a, b) \in [-c_2\sqrt{n}, c_2\sqrt{n}]$. The $2s$ players want to determine if $\Delta(a, b) \ge c_1\sqrt{n}$ or $\Delta(a, b) \le  -c_1\sqrt{n}$. Here $c_1<c_2$ are both constants. (Similarly, if $\Delta(a, b)$ is between $- c_1\sqrt{n}$ and $c_1\sqrt{n}$, an arbitrary output is allowed).
\end{definition}

To prove the $\Omega(sn)$ lower bound, we use a similar symmetrization augment as in~\cite{PVZ16} and reduce to the \textsf{GHD} problem. For the reduction, we consider $s = 4t + 2$ for simplicity, and without loss of generality by padding, and consider the following distribution $\mu$ for the inputs $a_i^j$ for players $j = 1, 2, \dots, 2t + 1$. 
Choose a uniformly random vector $a\in \{-1,1\}^n$. For each $i$, if $a_i = 1$, we place $(t + 1)$ bits of $1$ and $t$ bits of $-1$ randomly among the $2t + 1$ players in this coordinate; if $a_i = -1$, we place $t$ bits of $1$ and $(t +1)$ bits of $-1$ randomly among the $2t + 1$ players. We remark that under this distribution,  each player’s inputs are drawn from the same distribution, and each coordinate of each player is $1$ with probability $1/2$ and $-1$ with probability $1/2$. The distribution of $b_i^j$ is the same as that of $a_i^j$ for players $j = 2t + 2, \dots, 4t + 2$.

\begin{theorem}
\label{thm:k_gap}
Any protocol that solves the $s$-GAP problem with large constant probability requires $\Omega(sn)$ bits of communication.     
\end{theorem}

\begin{proof}
\looseness=-1 We reduce the $s$-GAP problem to the \textsf{GHD} problem using a similar symmetrization argument to that in~\cite{PVZ16}. Alice picks a random number $i \in [2t + 1]$ uniformly and simulates the $i$-th player. Bob simulates the remaining $s - 1$ players. We shall show that if there is an $s$-player protocol solving the $s$-GAP problem, then the coordinator will be able to solve the \textsf{GHD} problem on a constant fraction of the input vectors $a$ and $b$, which requires $\Omega(n)$ bits of communication. Note that the input distribution of each player is the same and Alice is choosing a random player. Hence, Alice's expected communication to Bob is at most $O(\chi/s)$ bits if the $s$-GAP problem can be solved using $\chi$ bits of communication, which yields a lower bound of $\Omega(sn)$ bits for the $s$-GAP problem.

We first consider Bob's information when he simulates $s - 1$ players. He knows each coordinate of $b$ directly. Consider a coordinate of $a$. If the sum of Bob's $s - 1$ bits on this coordinate is $2$ or $-2$, then he knows Alice's bit on this coordinate immediately, as their sum should be $1$ or $-1$; while if Bob's sum is $0$, he has zero information about Alice's bit on this coordinate. By a simple computation, we obtain that Bob's sum is $2$ or $-2$ with probability $\frac{t}{2t + 1}$ and is $0$ with probability $\frac{t +1}{2t + 1}$.
From a Chernoff bound, we see that with probability at least $1 - e^{-\Omega(n)}$, Bob learns at most $\frac{3}{5}n$ coordinates of $a$. Let $I$ denote the set of remaining indices. Then $|I| \ge \frac{2n}{5}$. We will show that Alice and Bob can solve \textsf{GHD} on $a_I$ and $b_I$ by simulating the protocol for the $s$-GAP problem.

Consider $\Delta(a_{J}, b_{J})$ for $J = [n] \setminus I$. With probability at least $99/100$, it will be contained in $[-c_1\sqrt{|J|}, c_1 \sqrt{|J|}]$, where $c_1$ is a sufficiently large absolute constant. Conditioned on this event, we have that whether the distance $\Delta(a_I, b_I) \ge c_2 \sqrt{|I|}$ or  $\Delta(a_I, b_I) \le -c_2 \sqrt{|I|}$ will decide whether $\Delta(a, b) \ge c_3 \sqrt{n}$ or $\Delta(a, b) \le -c_3 \sqrt{n}$, where $c_2,c_3>0$ are appropriate constants (recall that we have $|I| \ge \frac{2}{5}n$ and $|J| \le \frac{3}{5}n$). This means that, by simulating a $2s$-player protocol for the $s$-GAP problem, Alice and Bob can solve the $\textsf{GHD}_{|I|,c_2}$ problem on $a_I$ and $b_I$, which requires $\Omega(|I|) = \Omega(n)$ bits of communication. 
\end{proof}

\begin{corollary}
Any protocol that solves $m$ independent copies of the $s$-GAP problem with high constant probability requires $\Omega(snm)$ bits of communication.     
\end{corollary}

\begin{proof}
    Similar to the proof of Theorem~\ref{thm:k_gap}, Alice and Bob in this case need to solve $m$ independent copies of \textsf{GHD}. The direct sum theorem~\cite{CSWY01,BJKS04} states that if the information cost of solving a communication problem with probability $2/3$ is $f$, then the information cost of solving $m$ independent copies of the same communication problem simultaneously with probability at least $2/3$ is $\Omega(mf)$. Since the information cost implies a communication lower bound, it follows from Lemma~\ref{lem:gap} and the direct sum theorem that $\Omega(knm)$ bits of communication are required. %
\end{proof}

\subsection{$\Omega(sd/\eps^2)$ and $\Omega(sd/\eps)$ Lower Bounds} 

In this section, we will show an $\Omega(sd/\eps^2)$ lower bound for the $\ell_p$-regression problem when $0 < p \le 1$ and an $\Omega(sd/\eps)$ lower bound when $1 < p\leq 2$. 

For simplicity, we first consider the case of $d = 1$ and will later extend the result to general $d$. Consider the same input distribution as in Definition~\ref{def:k_gap} with $n = 1/\eps^2$, and for which the $2s$ players want to compute a $(1 + \eps)$-approximate solution to the $\ell_p$ regression problem 
\begin{equation}
\label{eq:1}
    \argmin_{x \in \R} \norm{ax - b}_p^p \;.
\end{equation}

In the lemma below, we shall show that using a $(1 + \eps)$-approximate solution for the $\ell_p$-regression problem~\eqref{eq:1}, the players can distinguish the two cases to the $s$-GAP problem for the vectors $a$ and $b$, which implies an $\Omega(s/\eps^2)$ lower bound. The proof, analogous to that of~\cite[Theorem~12.2]{MMW+22}, analyzes an objective of the form $r|1 - x|^p + (n - r)|1 + x|^p$ for $r=(n+\Delta(a,b))/2$.

\begin{lemma}\label{lem:auxiliary_d=1}
    Suppose that $p\in (0,2]$, $n = \Theta(1/\eps^2)$, and $a$ and $b$ are the vectors drawn from the distribution in Definition~\ref{def:k_gap}. Let $\eta = \eps$ when $p\in (0,1]$ and $\eta = \eps^2$ when $p\in (1,2]$. Then, any $\tilde{x}$ such that 
    $
    \|a\tilde{x} - b\|_p^p \le (1 + \eta)\min_{x \in \R} \|ax - b\|_p^p
    $
    can be used to distinguish whether $\Delta(a, b) \ge c\sqrt{n}$ or $\Delta(a, b) \le  -c\sqrt{n}$, where $c$ is an absolute constant.
\end{lemma}

\begin{proof}
    Suppose that $a_i=b_i$ for $r$ coordinates $i$ and $a_i\neq b_i$ for $n - r$ coordinates $i$. The objective function $\norm{ax-b}_p^p$ can be rewritten as 
    \[
    r\cdot|1 - x|^p + (n - r) \cdot |1 + x|^p \;.
    \]
    
    \paragraph{Case $p\in (0,1)$.   } The first observation is that the optimal solution $x^*$ should lie in $[-1, 1]$, otherwise $x = 1$ or $x = -1$ will give a lower cost. Next, without loss of generality, we can assume that $\Delta(a, b) \ge c\sqrt{n}$, which means that $r \ge \frac{n}{2} + \frac{c}{2}\sqrt{n}$. Following a similar analysis to that in~\cite[Theorem~12.2]{MMW+22}, we can now obtain that the optimal solution $x^*$ satisfies $x^* > 0$ and every $x < 0$ will lead to $\|ax-b\|_p^p \geq (1 + \eps)\|ax^* - b\|_p^p$. The case where $\Delta(a, b) \le  -c\sqrt{n}$ is similar, where the optimal solution $x^*$ satisfies $x^* < 0$ and every $x > 0$ will lead to $\|ax-b\|_p^p \geq (1 + \eps)\|ax^* - b\|_p^p$. Hence, using the sign of $x$ and the fact that $x$ is a $(1+\eps)$-approximate solution, we can distinguish the two cases of $\Delta(a,b)$.

    \paragraph{Case $p = 1$.} The objective can now be rewritten as 
    \[
    r\cdot|1 - x| + (n - r) \cdot |1 + x| \;.
    \]
    Without loss of generality, we assume that $\Delta(a, b) \ge c\sqrt{n}$ which means that $r \ge \frac{n}{2} + \frac{c}{2}\sqrt{n}$. The only thing we have to show is that $\|ax-b\|_p^p \geq (1 + \eps)\|ax^* - b\|_p^p$ for all $x < 0$. On the one hand, we have that $\norm{ax^* - b}_p^p \le \norm{a\cdot 1 - b}_p^p \le n - c\sqrt{n}$. On the other hand, when $x < 0$, noting that $r > n - r$,  we have that $\norm{ax - b}_p^p \ge \norm{a\cdot 0 - b}_p^p = n\geq (1+\eps)(n-c\sqrt{n})$. The last inequality follows from our choice of $n = \Theta(1/\eps^2)$. To conclude, when $p = 1$, we can also distinguish the two cases from the sign of $x$.

    \paragraph{Case $p \in (1,2)$.} The case of $1 < p < 2$ was shown in \cite[Theorem 12.4]{MMW+22}. Similar to their analysis, we can get that (i) when $\Delta(a, b) \ge c\sqrt{n}$, the optimal solution $x^*$ satisfies $x^* > 0$ and any $x < 0$ will yield $\|ax-b\|_p^p \geq (1+2\eps^2)\|ax^* - b\|_p^p$; 
    (ii) when $\Delta(a, b) \le  -c\sqrt{n}$, the optimal solution $x^*$ satisfies $x^* < 0$ and 
    any $x > 0$ will yield $\|ax-b\|_p^p \geq (1+2\eps^2)\|ax^* - b\|_p^p$. 
    Hence, we can deduce the sign of $x$ in the two cases, and can distinguish the two cases 
    when $x$ is a $(1+\eps^2)$-approximate solution.

    \paragraph{Case $p = 2$.} 
    The optimal solution is $x^* = \frac{\sum_i a_i b_i}{\sum_i a_i^2} = \frac{\sum_i a_i b_i}{n}$ and the corresponding objective value is $n - \frac{(\sum_i a_ib_i)^2}{n}$. When $\Delta(a, b) \ge c\sqrt{n}$, the optimal solution $x^* > 0$ and $\|ax^* - b\|_2^2 \leq n - c^2$, while for all $x < 0$, from the property of the quadratic function, we get that $\norm{ax^* - b}_2^2 \ge \norm{a\cdot(0) - b}_2^2 = n \geq (1+2\eps^2)(n-c^2)$ (recall that $n \leq c/(2\eps^2)$).
    A similar analysis works when $\Delta(a, b) \le -c\sqrt{n}$ and the proof is complete.
\end{proof}

Combining this lemma with Theorem~\ref{thm:k_gap} yields the desired lower bound for the distributional regression problem with $d=1$.

\begin{lemma}
    Suppose that $d = 1$ and $\eps > 0$. Then any protocol that computes a $(1 + \eps)$-approximate solution to the $s$-server distributional $\ell_p$-regression problem in the message passing model with high constant probability requires $\Omega(s/\eps^2)$ bits of communication for $p\in (0,1]$ and $\Omega(s/\eps)$ bits of communication for $p\in (1,2]$.
\end{lemma}

We now extend the lower bound to general $d$ via a padding argument. Suppose that $a_1, a_2, \dots, a_d$ and $b_1, b_2, \dots, b_d$ are $d$ independent samples drawn from the same distribution as defined in Definition~\ref{def:k_gap} with $n = \Theta(1/\eps^2)$. We form a matrix $A \in \R^{O(d/\eps^2) \times d}$ and a vector $b \in \R^{O(d/\eps^2)}$ as
\[
A = \begin{bmatrix}
a_1 & & \\
& a_2 & & \\
& & \ddots & \\
& & & a_d \\
\end{bmatrix}, \ \ \ b = \begin{bmatrix}
    b_1 \\
    b_2 \\
    \vdots \\
    b_d \\
\end{bmatrix}
\;.
\]
It then follows that
\[
\min_{x \in \R^d}\norm{Ax - b}_p^p = \sum_{i = 1}^{d} \min_{x_i \in \R} \norm{a_i x_i - b_i}_p^p.
\]
We then make the following observation. If $x \in \R^d$ is a $(1 + \eps)$-approximate solution of $\min_x \norm{Ax - b}_p^p$, then there must exist 
 a constant fraction of the indices $i \in [d] $ such that $x_i$ is a $(1 + O(\eps))$-approximate solution to the regression problem $\min_{x_i \in \R} \norm{a_i x_i - b_i}_p^p$ (recall that we have the guarantee that $\Delta(a_i, b_i) \in [-c_2\sqrt{n}, c_2\sqrt{n}]$ for all $i$, and hence the objective values for each regression problem are within a constant factor). This means that from the signs of these $x_i$, we can solve a constant fraction of the $d$ independent copies of the $s$-GAP problem, which implies the following theorem immediately.

 \begin{theorem}
 \label{thm:lb_d/eps}
     Suppose that $\eps > \frac{1}{\sqrt{n}}$ for $p \in (0,1]$ and $\eps > \frac{1}{n}$ for $p \in (1,2]$. Then any protocol that computes a $(1 + \eps)$-approximate solution to the $s$-server distributional $\ell_p$-regression problem with $d$ columns in the message passing model with large constant probability requires $\Omega(sd/\eps^2)$ bits of communication for $p\in (0,1]$ and $\Omega(sd/\eps)$ bits of communication for $p\in (1,2]$.
 \end{theorem}

\subsection{$\Omega(sd^2)$ Lower Bound for $p \in (0, 2]$}

    In this section, we present an $\Omega(sd^2)$ lower bound for $0 < p \le 2$. We first describe the intuition behind our lower bound. Following~\cite{VWW20}, we construct a set of matrices $\mathcal{H} \subseteq \mathbb{R}^{d \times d}$ with a vector $b \in \mathbb{R}^d$ such that (i) $T$ is non-singular for all $T \in \mathcal{H}$, and (ii) $S^{-1} b \ne T^{-1}b$ for all $S, T \in \mathcal{H}$ and $S\neq T$. Then we uniformly sample a matrix $A \in \mathcal{H}$ and show that we can obtain the index of $A$ in the set $\mathcal{H}$ from a constant-factor approximate solution to the regression problem $\min \norm{Ax - b}_p^p$. This will imply an $\Omega(d^2)$ lower bound even for $s = 2$. The construction of $\mathcal{H}$ is given in the following lemma.

\begin{lemma}
\label{lem:set}
    For every sufficiently large $d$, there exists a set of matrices $\mathcal{H} \subseteq {\{-1,1\}}^{d \times d}$ with $|\mathcal{H}| = \Omega(2^{0.49d^2})$ such that (i)
 $T$ is non-singular for all $T \in \mathcal{H}$, and (ii) for all distinct $S, T \in \mathcal{H}$, $S^{-1} e_d \neq T^{-1} e_d$, where $e_d$ is the $d$-th standard basis vector. 
\end{lemma}

We remark that in~\cite{VWW20}, Lemma~\ref{lem:set} was only shown for the case where $t > 1$, $|\mathcal{H}| = \Omega(t^{1/6d^2})$ and the matrix entries are integers in $[-t, t]$. However, using the singularity probability of random matrices in $\{-1, +1\}^{d \times d}$ and following a similar argument to~\cite{VWW20}, we can obtain the desired bounds in Lemma~\ref{lem:set}. The detailed proof can be found in Appendix~\ref{sec:set}. Note that the construction procedure of the set is close to random sampling -- uniformly sample $\Omega(2^{0.49d^2})$ matrices and remove a small fraction. This property will be  crucial to our proof.


To achieve an $\Omega(sd^2)$ lower bound for $s$ players, we consider the same input distribution for the $s$ players in Lemma~\ref{lem:gap} and employ a similar symmetrization argument. After sampling matrices in $\mathcal{H}$, we construct the inputs of the $s$ players to be matrices in $\{-1, +1\}^{d \times d}$ with the sum being $A$. However, if we follow the same argument and let Bob simulate $s - 1 = 2t$ players, in expectation he will know a $\frac{t}{2t + 1} \approx \frac{1}{2}$ fraction of the entries of $A$, and from the construction of the set $|\mathcal{H}|$ we know that there will be only $O(1)$ matrices in $\mathcal{H}$ satisfying the conditions on such entries. Hence, Alice only needs to send $O(1)$ bits of information to Bob. To solve this issue, we make the following modification. Instead, we let Alice simulate $2$ players, and Bob simulates the remaining $s - 2 = 2t - 1$ players. In this case, Bob will know roughly a $1/4$-fraction of the entries directly; however, for the remaining entries, he will know side information. Roughly speaking, for $A_{ij}$, if Bob's sum over the $s - 2$ players is $1$, with probability roughly $2/3$, $A_{ij}$ is $1$; if his sum over the $k - 2$ players is $-1$, with probability roughly $2/3$, $A_{ij}$ is $-1$. We shall show that even having such side information, with high probability the conditional entropy of the remaining entries of $A$ is still $\Omega(d^2)$, which implies that Alice still needs to send Bob $\Omega(d^2)$ bits. 



\begin{lemma}\label{lem:sd^2-game}
    Consider the following game of $s = 2t + 1$ players, where the $i$-th player receives a $d \times d$-matrix $A^i$ such that $A^i\subseteq \{-1,1\}^{d\times d}$ with the guarantee that $A = \sum_{i} A^i$ is distributed in $\mathcal{H}$ uniformly. The $s$ players want to determine collectively the index of the matrix $A$ in $\mathcal{H}$.
    Any protocol which solves this problem with large constant probability requires $\Omega(sd^2)$ bits of communication.
\end{lemma}

\begin{proof}
    We first describe the input distribution of each player. Suppose that matrix $A$ has been sampled from $\mathcal{H}$. For each coordinate $(i,j)$, if $A_{ij} = 1$, we place $(t + 1)$ bits of $1$ and $t$ bits of $-1$ randomly among the $2t + 1$ players' inputs for coordinate $j$; if $A_{ij} = -1$, we place $t$ bits of $1$ and $t +1 $ bits of $-1$. Similarly, under this distribution,  each player’s inputs are drawn from the same distribution. 

    We then use symmetry and let Alice simulate two random players, and Bob simulates the remaining $s - 2 = 2t - 1$ players. Consider first Bob's information when he simulates $2t - 1$ players. Via a simple computation we can get that for each coordinate, with probability $\frac{t - 1}{4t + 2}$ Bob's sum will be $3$ or $-3$, in which case he will know $A_{ij}$ immediately. If Bob's sum is $1$, he will get that $A_{ij}=1$ with probability $\frac{2}{3}$  and $A_{ij}=-1$ with probability $\frac{1}{3}$; if Bob's sum is $-1$, he will get that $A_{ij}=-1$ with probability $\frac{2}{3}$ and $A_{ij}=1$ with probability $\frac{1}{3}$. It follows from a Chernoff bound that with probability $1 - \exp(-d^2)$, Bob obtains the exact information of at most $0.26d^2$ coordinates and has partial information about the remaining coordinates. For the remainder of the proof we assume this event happens.

    Let $\mathcal{S}$ denote the subset of $\mathcal{H}$ which agrees on the above $0.26d^2$ coordinates. From the construction of $\mathcal{H}$ we get that with at least constant probability $|\mathcal{S}| = \Omega(2^{0.2 d^2})$. Condition on this event. For simplicity, next we only consider the matrix in $\mathcal{S}$ and treat it as an $\ell$-dimensional vector after removing the known $0.26d^2$ coordinates, where $\ell = 0.74d^2$. Let $Y$ denote Bob's sum vector. We shall show that the conditional entropy $H(A \ | \ Y)$ remains $\Omega(d^2)$, and hence by a standard information-theoretic argument, Alice must still send $\Omega(d^2)$ bits to Bob to identify the index of the matrix in $\mathcal{S}$. From this, we get an $\Omega(sd^2)$ lower bound on the protocol for the original problem.

    By a Chernoff bound, with probability $1 - \exp(-d^2)$, the Hamming distance between $A$ and $Y$ is within $\frac{1}{3}\ell \pm 0.01d^2$. We condition on this in the remainder of the proof. We now turn to bound the number of matrices in $S$ which have a Hamming distance of $\frac{1}{3}\ell$ from $Y$. For each matrix $B$, from the construction of $\mathcal{H}$ we know that each coordinate of $B$ is the same as the corresponding coordinate of $A$ with probability $1/2$. Hence, the probability that $B$ has Hamming distance $\frac{2}{3}\ell$ from $A$ is (using Stirling's formula)
    \[
    \binom{\ell}{\frac{2}{3}\ell} \cdot 2^{-\ell} \simeq \frac{1}{\ell}\cdot \frac{3^\ell}{2^{\frac{2}{3}\ell}} \cdot 2^{-\ell} = \frac{3^\ell}{\ell\ 2^{\frac{5}{3}\ell}}.
    \]
    Hence, the expected number of such $B$ is 
    \[
    |\mathcal{S}| \cdot \frac{3^\ell}{\ell 2^{\frac{5}{3}\ell}} > 2^{0.2 d^2} \cdot \frac{3^\ell}{\ell 2^{\frac{5}{3}\ell}} \ge (1.101)^{d^2} \;.
    \]
    From a Chernoff bound we know that with probability at least $1 - \exp(-d^2)$, the number of $B \in \mathcal{S}$ for which $B$ has a Hamming distance $\frac{1}{3}\ell$ from $Y$ is at least $(1.10)^{d^2}$.  
    
    We next turn to show that when conditioned on the event above, it is enough to show that the conditional entropy $H(A\mid Y)$  satisfies $H(A \mid Y) = \Omega(d^2)$ given Bob's vector $Y$. Let $\mathcal{T}$ be the subset of $\mathcal{H}$ which agrees on the above $0.26d^2$ coordinates and having Hamming distance within $\frac{1}{3}\ell \pm 0.01d^2$. For each matrix $T \in \mathcal{T}$, define a weight of the matrix $T$ to be $w_T = \left(\frac{2}{3}\right)^{\ell - u}\left(\frac{1}{3}\right)^{u} = (\frac{1}{3})^\ell 2^{l- u}$, where $u$ is the Hamming distance between $T$ and $Y$. It follows from Bayes' Theorem that $T$ is the correct matrix with probability
    \[
    p_T = \frac{w_T}{\sum_{i \in \mathcal{T}} w_i} \;.
    \]
    For the denominator, we have from the conditioned events that 
    \[
    S = \sum_{i \in \mathcal{T}} w_i \ge (1.10)^{d^2} \cdot \left(\frac{1}{3}\right)^{\ell} 2^{\frac{2}{3}\ell - 0.01d^2}   \ge (0.682)^{d^2} \;.
    \]
    For the numerator, note that it holds for every $i \in \mathcal{T}$ that 
    \[
    w_i \le \left(\frac{1}{3}\right)^\ell 2^{\frac{2}{3}\ell + 0.01d^2} \le (0.629)^{d^2}.
    \]
    It follows from the definition of the entropy that
    \[
    H(A \mid Y) = \sum_{i \in \mathcal{T}} p_i \log \frac{1}{p_i} = \sum_{i \in \mathcal{T}} \frac{w_i}{S} \log\frac{S}{w_i} 
    \ge \sum_{i \in \mathcal{T}}\frac{w_i}{S}\log \frac{S}{(0.629)^{d^2}}  
    = \log \frac{S}{(0.629)^{d^2}}  = \Omega\left(d^2\right) \;,
    \]
    which is exactly we need. The proof is complete.
\end{proof}

The following theorem follows immediately from the preceding lemma.

\begin{theorem}
\label{thm:lb_sd^2}
    Suppose that $0 < p \le 2$. Any protocol that computes a constant-factor approximate solution to the $s$-server distributional $\ell_p$-regression problem with $d$ columns in the message passing model with large constant probability requires $\Omega(sd^2)$ bits of communication. 
\end{theorem}

\section{$\ell_2$-Regression Upper Bound}

In this section, we give an $\tilde{O}(sd^2 + sd/\eps)$ communication protocol for the distributed $\ell_2$-regression problem. We first describe the high-level intuition of our protocol, which is based on the sketching algorithm in~\cite{CW09} and the sketching-based pre-conditioning algorithm in~\cite{CW17}.
\begin{itemize}[nosep]
    \item Let $S_1 \in \R^{O(d\log(d)/\eps) \times n}$ be a $(1 \pm \sqrt{\eps})$-subspace embedding. We compute $\hat{A} = SA$ and $\hat{b} = Sb$ and then the problem is reduced to solving $\min_{x \in \R^d} \norm{\hat{A}x - \hat{b}}_2^2$.
    \item Let $S_2 \in \R^{O(d \log d) \times O(d\log(d)/\eps)}$ be a $(1 \pm 1/2)$ subspace embedding of $SA$. We compute a QR-decomposition of $S\hat{A} = QR^{-1}$. Then the regression problem is equivalent to solving $\min_{x \in \R^d} \norm{\hat{A}Rx - \hat{b}}_2^2$.
    \item Run a gradient descent algorithm for $T = O(\log(1/\eps))$ iterations. In the $t$-th iteration, compute the gradient of the objective function at the current solution $x_t$ and perform the update $x_{t + 1} = x_t - (\hat{A}R)^T(\hat{A}Rx_t - \hat{b})$.
    \item Output $Rx_T$ as the solution.
\end{itemize}

\medskip

The protocol is presented in Algorithm~\ref{alg:ell_2}. Initially, each server computes $\hat{A}^i = \Pi_2 \Pi_1 A^i$, then computes $\Pi_3 \hat{A}^i$ and sends it to the coordinator. Note  that $\Pi_1$ is a $\textsf{Count-Sketch}$ matrix and hence we can compute $\Pi_1 A^i$ in $\nnz(A^i)$ time and then compute $\Pi_2 \Pi_1 A^i$ in $\nnz(A^i) + \poly(d/\eps)$ time. The coordinator then computes a QR-decomposition of $\Pi_3 \hat{A} = \sum_i \Pi_3  \hat{A}^i$. The point is that $\hat{A}R$ will be well-conditioned, which will greatly improve the convergence rate of gradient descent. Then each server will help compute the gradient at the current solution $x_t$ and the coordinator will perform the corresponding update. The following is our theorem. 

\begin{algorithm}[tb]
\begin{mdframed}
\begin{enumerate}[nosep]
\item Each Server $P_i$ initializes the same \textsf{Count-Sketch} matrix $\Pi_1 \in \R^{O(d^2/\eps) \times n}$ and \textsf{OSNAP} matrices $\Pi_2\in \R^{O(d (\log d) / \eps) \times O(d^2/\eps)}$, and $ \Pi_3\in \R^{O(d \log d) \times O(d(\log d) / \eps)}$.
\item Each Server $P_i$ computes $\hat{A}^i = \Pi_2 \Pi_1 A^i$ and $\hat{b}^i = \Pi_2 \Pi_1 \hat{b}^i$.
\item Each Server $P_i$ computes $\Pi_3 \hat{A}^i$ and sends it to the coordinator.
\item The coordinator computes a $QR$ decomposition of $\Pi_3 \hat{A} = \sum_i \Pi_3 \hat{A}^i= QR^{-1}$ and sends $\tilde{R}$ to each server $P_i$, where 
$\tilde{R}$ satisfies that (i) every entry of $\tilde{R}$ is an integer multiple of $1/\poly(nd)$, (ii) every entry of $\tilde{R} - R$ is in $[-1/\poly(nd), 1/\poly(nd)]$ and (iii) $\tilde{R}$ is invertible.
\item Each server initializes $x_1 = \mathbf{0}^d$. For $t = 1, 2, \ldots, T = O(\log(1/\eps))$
\begin{enumerate}[nosep]
    \item Each server computes $\hat{A}^i \tilde{R} x_t - \hat{b}^i$ and sends it to the coordinator.
    \item The coordinator computes $y_t = \hat{A}\tilde{R}x_t - \hat{b} = \sum_i (\hat{A}^i \tilde{R} x_t - \hat{b}^i)$ and sends it to each server.
    \item Each server computes $(\hat{A}^i)^T y$, and sends it to the coordinator. The coordinator computes $g_t = B^T y = \sum_i (\hat{A}^i)^T y$, and makes the update $x_{t + 1} = x_t - g_t$, then sends $x_{t + 1}$ to each server.
    \item The coordinator computes $\tilde{R}x_T$ as the solution.
\end{enumerate}
\end{enumerate}
\end{mdframed}
\caption{Protocol for $\ell_2$ regression in the message passing model}
\label{alg:ell_2}
\end{algorithm}

\begin{theorem}\label{thm:ell_2_regression}
    The protocol in Algorithm~\ref{alg:ell_2} returns a $(1 \pm \eps)$-approximate solution to the $\ell_2$-regression problem with large constant probability, and the communication complexity is $\tilde{O}(sd^2 + sd/\eps)$. Moreover, the total runtime of all servers of the protocol is $O(\sum_i \mathrm{nnz}(A^i) + s \cdot \poly(d/\eps))$. 
\end{theorem}

To prove the correctness of Algorithm~\ref{alg:ell_2}, we need the following lemmas. The reader can find more detail in~\cite{woo14}.

\begin{lemma}
    \label{lem:sqrt}
    Suppose that $S$ is a $(1 \pm \sqrt{\eps})$-subspace embedding and $x'= \argmin_{x \in \R^d} \norm{S(Ax - b)}_2$. Then it holds with large constant probability that
    \[
    \norm{Ax' - b}_2 \le (1 + \eps) \norm{Ax - b}_2 \;. 
    \]
    Further suppose that $x_c$ is a $(1 + \eps)$-approximate solution to $\min_{x \in \R^d} \norm{S(Ax - b)}_2$, it then holds that
    \[
    \norm{S(Ax_c - b)}_2 \le (1 + \eps) \norm{Ax - b}_2 \;. 
    \]
\end{lemma}


We remark that the case where $x_c$ is the minimizer was shown by \cite{CW09} and the case where $x_c$ is a $(1 +\eps)$-approximate solution was recently shown by \cite{MWZ22}. 

\begin{lemma}
\label{lem:ell_2_error}
    Suppose that $S$ is a $(1 \pm \eps_0)$-subspace embedding and consider the iterative algorithm above, then
    \[
    \norm{\hat{A}R{x_{t + 1} - x^*}}_2 = \eps_0^m \cdot \norm{\hat{A}R{x_t - x^*}}_2 \;.
    \]
    As a corollary, when $t = \Omega(\log(1/\eps))$, it holds that $\norm{\hat{A}Rx_t - \hat{b}}_2^2 \le (1 + \eps) \norm{\hat{A}Rx^* - \hat{b}}_2^2$.
\end{lemma}

Now we are ready to prove Theorem~\ref{thm:ell_2_regression}.

\begin{proof}[Proof of Theorem~\ref{thm:ell_2_regression}]
    Since $\Pi_1$ has $O(d^2/\eps)$ rows and $\Pi_2$ has $O(d\log(d)/\eps)$ columns, from Section~\ref{sec:prelim} we get that with probability at least $99/100$, both $\Pi_1$ and $\Pi_2$ are $(1 \pm O(\sqrt{\eps}))$ subspace embeddings, which means $\Pi_2\Pi_1$ is a $(1 + \sqrt{\eps})$-subspace embedding. 

     Let $\hat{A} = \Pi_2 \Pi_1 A$ and $\hat{b} = \Pi_2 \Pi_1 b$. From Lemma~\ref{lem:sqrt}, we see that it suffices to solve $\min_{x \in \R^d} \norm{\hat{A}x - \hat{b}}_2$.
%
    Conditioned on these events, it follows immediately from Lemma~\ref{lem:ell_2_error} that $x_T$ is a $(1 \pm \eps)$-approximate solution to $\min_{x \in \R^d} \norm{\hat{A}x - \hat{b}}_2$, provided that each server uses $R$ instead of $\tilde{R}$. To show that $\tilde{R}$ works here, note that an initial step in the proof of Lemma~\ref{lem:ell_2_error} is that $\norm{S\hat{A}Rx}_2 = 1$ for all unit vectors $x$, which implies that $\norm{\hat{A}Rx}_2 \in [1 - \eps_0, 1 + \eps_0]$. For $\tilde{R}$, we have that 
    \[
    \norm{S\hat{A}Rx}_2 - \norm{S\hat{A}\tilde{R}x}_2 \le \norm{S\hat{A}(R - \tilde{R})x}_2 
    \le 2 \norm{\hat{A}}_2 \norm{(R - \tilde{R})x}_2 
    \le 1/\poly(nd) \;.
    \]
    The last inequality is due to the fact that each entry of $R - \tilde{R}$ is $O(1/\poly(nd))$ and each entry of $\hat{A}$ is $O(\poly(nd))$. Hence, $\norm{ARx} \in [1 - 1.1\eps_0, 1 + 1.1\eps_0]$ will still hold and a similar argument will go through, yielding that $x_T$ is a $(1 \pm \eps)$-approximate solution.

    We next analyze the communication complexity of the protocol. For Step~$3$, since $\Pi_3 \hat{A}^i$ is an $O(d \log d) \times d$ matrix, each server $P_i$ sends $\tilde{O}(d^2)$ entries. Each entry of $A^i$ has magnitude $[1/n^c, n^c]$, and thus each entry of $\Pi_1 A^i$ is contained in $[1/n^{c},n^{c+1}]$, each entry of $\hat{A}^i = \Pi_2 \Pi_1 A^i$ is contained in $[\eps/n^{c+2},n^{c+3}/\eps]$ and each entry of $\Pi_3\hat{A}^i$ is contained in $[\eps^2/n^{c+4},n^{c+5}/\eps^2]$, which implies that each entry of $\Pi_3 \hat{A}^i$ can be described using $O(\log(n/\eps))$ bits and thus a total communication of $O(sd^2)$ bits for Step~3. In Step~$4$, since $\tilde{R}$ is a $d \times d$ matrix and each entry is an integer multiple of $1/\poly(nd)$, the coordinator sends $\tilde{R}$ to each server using $\tilde{O}(sd^2)$ bits in total. In each iteration of Step~$5$, we note that $y_t$ is an $O(d/\eps)$-dimensional vector and $g_t$ is a $d$-dimensional vector, and each of their entries has $O(\log(nd))$ precision. Hence, the total communication of each iteration is $\tilde{O}(sd/\eps)$. Putting everything together, we conclude that the total amount of the communication is $\tilde{O}(sd^2 + \log(1/\eps) \cdot (sd/\eps)) = \tilde{O}(sd^2 + sd/\eps)$ bits. 

    We now consider the runtime of the protocol. To compute $\Pi_2 \Pi_1 A^i$, notice that $\Pi_1$ is a \textsf{Count-Sketch} matrix, and hence each server takes $\nnz(A^i)$ time to compute $\Pi_1 A^i$ and then use $\poly(d/\eps)$ time to compute $\Pi_2(\Pi_1 A^i)$. Hence, Step 2 takes $O(\sum_i \nnz(A^i))$ time. For the remaining steps, one can verify that each step takes $\poly(d/\eps)$ time on a single server or on the coordinator. The total runtime is therefore $O(\sum_i \mathrm{nnz}(A^i) + s \cdot \poly(d/\eps))$.
\end{proof}

\section{$\ell_p$-Regression Upper Bound}

In this section, we give an $\tilde{O}(sd^2/\eps + sd/\eps^{O(1)})$ communication protocol for the distributed $\ell_p$-regression problem when $1 < p < 2$. We first describe the high-level intuition of our protocol.

\begin{itemize}[nosep]
    \item Let $T \in \R^{O(d(\log d)/\eps^{O(1)}) \times n}$ be a sketch matrix whose entries are scaled i.i.d. $p$-stable random variables. We compute $\hat{A} = TA$ and $\hat{b} = Tb$ and then the problem is reduced to solving $\min_{x \in \R^d} \norm{\hat{A}x - \hat{b}}_r$.
    \item Run Algorithm~\ref{alg:lewis} to obtain a constant approximation of the $\ell_r$ Lewis weights $w$ of $[\hat{A} \ \hat{b}]$.
    \item Sample $O(d/\eps)$ rows of $\hat{A}$ and $\hat{b}$ proportional to $w$, and form the new matrix $A'$ and $b'$. 
    \item Solve $x = \argmin_{x \in \R^d} \norm{A' x - b'}_r$ and output $x$. 
\end{itemize}

\medskip

The protocol is shown in Algorithm~\ref{alg:ell_p}. To show its correctness, we first analyze $\ell_p$-to-$\ell_r$ embeddings and the algorithm for solving the $\ell_p$-regression problem using Lewis weight sampling.

\begin{algorithm}[t]
\begin{mdframed}
\begin{enumerate}[nosep] 
\item Each server initializes the same $p$-stable variable matrix $T \in \R^{ \log(d)/\eps^{O(1)} \times n}$, \textsf{OSNAP} matrix $S \in \R^{d\log(d) \times d\log(d)/\eps^{O(1)}}$ and Gaussian matrix $G \in \R^{(d+1) \times \log(d/\eps)}$. The entries of $T$ and $G$ are rounded to the nearest integer multiples of $1/\poly(nd)$.  
\item Each server $P_i$ computes $\hat{A}^i = T A^i$ and $\hat{b}^i = T b^i$, and forms $B^i = [\hat{A}^i \ \hat{b}^i]$. 
\item Each server initializes $w = \mathbf{1}^{d/\eps^{O(1)}}$. For $j = 1, 2, \ldots, t = O(\log\log(d/\eps))$
\begin{enumerate}[nosep] 
    \item Each server $P_i$ computes $S_t \tilde{W}^{1/2 - 1/r}B^i$ (where $W = \diag(w)$ and $\tilde{W}^{1/2 - 1/r}$ is a rounded version of $W^{1/2 - 1/r}$) and then sends it to the coordinator. 
    
    \item The coordinator computes the QR-decomposition of $S \tilde{W}^{1/2 - 1/r}B = QR^{-1}$. It then 
    sends $\tilde{R}$ to each server, where
$\tilde{R}$ satisfies that (i) every entry of $\tilde{R}$ is an integer multiple of $1/\poly(nd)$, (ii) every entry of $\tilde{R} - R$ is in $[-1/\poly(nd), 1/\poly(nd)]$ and (iii) $\tilde{R}$ is invertible.
    \item Each server computes $B^i \tilde{R} G$ and sends it to the coordinator.
    \item The coordinator computes the square of the $\ell_2$ norm of the rows in $B\tilde{R}G$ as a vector $\tau \in \R^{d/\eps^{O(1)}}$.
    \item The coordinator performs $w_i \leftarrow (w^{2/r - 1}\tau_i)^{r/2}$ and sends the new $w$ to all servers, after rounding each coordinate of $w$ to the nearest integer multiple of $1/\poly(nd)$.
\end{enumerate}
    \item The coordinator samples the $i$-th row of $\hat{A}$ and $\hat{b}$ with probability $q_i \ge \beta \cdot w_i \cdot \log^3(d/\eps) / \eps$, where $\beta$ is a sufficiently large constant. Suppose that $\mathcal{S}$ is the set of indices of the sampled rows. Each server sends the rows in $\mathcal{S}$ to the coordinator.
    \item The coordinator forms the matrix $A'$ and $b'$ using the rows in $\mathcal{S}$ and each sampled row with a re-scaling factor of $1/q_i^{1/r}$. 
    \item The coordinator solves $x = \argmin_{x \in \R^d} \norm{A'x- b'}_r$ and returns the solution $x$. 
\end{enumerate} 
\end{mdframed}
\caption{Protocol for $\ell_p$ regression in the message passing model}
\label{alg:ell_p}
\end{algorithm}

\paragraph{$p$-stable distribution.} The best known $(1 \pm \eps)$ $\ell_p$ subspace embeddings require an exponential number of rows for a $p$-stable sketch. However, as we will show in the following lemma, for $1 < r < p$, $\tilde{O}(d/\eps^{O(1)})$ rows are enough to give a $(1 \pm \eps)$ (lopsided) embedding from $\ell_p$ to $\ell_r$, which is sufficient for the regression problem.

\begin{lemma}
\label{lem:p_stable_embedding}
Suppose that $p>r>1$ are constants, and $T \in \R^{m \times n}$ is a matrix whose entries are i.i.d.\ $p$-stable random variables scaled by $1/(m^{1/r} \cdot \alpha_{p, r})$, where $\alpha_{p, r}$ is a constant depending on $p$ and $r$ only. For $m = d \log d/\eps^{C(\eps,r)}$, where $C(\eps,r)$ is a constant depending on $p$ and $r$ only, it holds for any given matrix $A\in \R^{n\times d}$ that
\begin{enumerate}[nosep]
    \item (dilation) for each $x\in\R^d$, $\norm{TAx}_r \le (1 + \eps) \norm{Ax}_p$ with large constant probability.
    \item (contraction) $\norm{TAx}_r \ge (1 - \eps) \norm{Ax}_p$ for all $x \in \R^d$ simultaneously with high probability.
\end{enumerate}
Furthermore, the entries of $T$ can be rounded to the nearest integer multiples of $1/\poly(nd)$ and the same guarantees still hold.
\end{lemma}

To prove the lemma, we need the following results.

\begin{lemma} [see, e.g., ~\cite{FG11}]
\label{lem:moment}     
    Suppose that $\alpha \in \R^d$ and $\theta \in \R^d$ is a vector whose entries are i.i.d. $p$-stable variables. Then it holds that
    \[
    \left(\E \abs{\sum_i \alpha_i \theta_i}^r \right)^{1/r} = \alpha_{p, r} \left(\sum_i |\alpha_i|^p \right)^{1/p}
    \]
    where $\alpha_{p, r}$ is a constant that only depends on $p$ and $r$.
\end{lemma}

\begin{proposition}\label{prop:exp}
    Suppose that $r,s\geq 1$ and $X$ is a random variable with $\E|X|^{rs} < \infty$. It holds that
    \[
    \E \abs{ \abs{X}^r - \E\abs{X}^r}^s \le 2^{s}\E |X|^{rs} \;.
    \]
\end{proposition}
\begin{proof}
    We have that
    \begin{align*}
        \E \abs{ \abs{X}^r - \E\abs{X}^r}^s &\leq 2^{s-1} \E\left( \abs{\abs{X}^r}^s + (\E\abs{X}^r)^s  \right) \\
        &\leq 2^{s-1} (\E\abs{X}^{rs} + (\E\abs{X}^r)^s) \\
        &\leq 2^{s-1}(\E\abs{X}^{rs} + \E\abs{X}^{rs}) \\
        &= 2^s \E\abs{X}^{rs}. 
    \end{align*}
\end{proof}

\begin{lemma}[{\cite[Theorem 2]{Von65}}]
\label{lem:sum_moment}
    Suppose that  $1 \le r \le 2$.
    Let $X_1,\dots,X_n$ be independent zero mean random variables with $\E[|X_i|^r] < \infty $. Then we have that 
    \[
    \E\left[\left(\sum_{i=1}^n |X_i|\right)^r\right] \le 2 \sum_{i=1}^n \E \left[\left|X_i\right|^r\right] \;.
    \]
\end{lemma}

\begin{lemma}
\label{lem:p_stable_upper_bound}
    Suppose that $p\in (1,2)$ is a constant and $T \in \R^{m \times n}$ is a matrix whose entries are i.i.d. $p$-stable entries scaled by $1/(\alpha_{p} \cdot m^{1/p} )$. For $m = d \log d/\eps^{O(1)}$, given any $A \in \R^{n \times d}$, it holds with large constant probability that for all $x \in \R^d$ 
    \[
    \norm{TAx}_p \le \poly(d) \norm{Ax}_p \;.
    \] 
\end{lemma}

We note that Lemma~\ref{lem:p_stable_upper_bound} was shown in~\cite{SW11} for $p = 1$. For $1 < p < 2$, a similar argument still goes through after replacing the $\ell_1$ well-conditioned basis with an $\ell_p$ well-conditioned basis.

\begin{proof}[Proof of Lemma~\ref{lem:p_stable_embedding}]
    First we consider the original $T$ without rounding the entries.

    Now we show $(1)$. Let $y = Ax$.
    From properties of $p$-stable random variables, we get that each $(Ty)_i$ follows the same distribution. From Lemma~\ref{lem:moment} we have that for every $i$, $\E|(Ty)_i|^r = \frac{\alpha_{p,r}^r}{\alpha_{p,r}^r \cdot m} \norm{y}_p^r = \frac{1}{m} \norm{y}_p^r$. To get concentration, we pick an $r'\in (r,p)$ and consider the $r'/ r$-moment of $(Ty)_i^r$.

    Similar to Lemma~\ref{lem:moment}, we have that $\mathbb{E} [|(Ty)|_i^{r'}] = \frac{\beta_{p,r,r'}}{m^{r'/r}}\norm{y}_p^{r'}$ is bounded, where $\beta_{p,r,r'}$ is a constant depending on $p,r,r'$ only. Let $S = \sum_i |(Ty)_i|^r$ and we have that $\E [S]= \norm{y}_p^r$. Consider the $(r/r')$-th moment of $S$. We then have 
    \begin{align*}
    \E \left[\left(S - \E[S]\right)^{r'/r}\right] & = \E \left[\left(\sum_i \left(|(Ty)_i|^r- \frac{1}{m} \norm{y}_p^r \right )\right)^{r'/r} \right] \\ 
    & \le 2  \left(\sum_i \E \left[ |(Ty)_i|^r- \frac{1}{m} \norm{y}_p^r|\right]^{r'/r}\right) \tag{Lemma~\ref{lem:sum_moment}}\\
    & \le 2^{r'/r+1} \left(\sum_i \E |(Ty)_i|^{r'} \right) \tag{Proposition~\ref{prop:exp}}\\
    & \le C \left(\sum_i \frac{1}{m^{r'/r}} \norm{y}_p^{r'} \right) \\
    &= C \norm{y}_p^{r'} / m^{r'/r - 1} \;,
    \end{align*}
    where $C$ is a constant that depends only on $r,r',$ and $p$. By Markov's inequality, we have that
    \begin{align*}
    \mathbf{Pr} \left[|S - \E[S]| \ge \eps \E[S]\right] &\le \mathbf{Pr} \left[|S - \E[S]|^{r'/r} \ge (\eps \E[S])^{r'/r}\right] \\
    & \le \frac{\E \left[\left(S - \E[S]\right)^{r'/r}\right]}{\eps^{r'/r} \norm{y}_p^{r'}} \\
    & \le \frac{C_{r'/r}}{\eps^{r'/r} m^{r'/r - 1}} \;.
    \end{align*}
    Hence, we can see that when $m = \Omega(1/\eps^{\frac{r'}{r' - r}}) = 1/\eps^{\Omega(1)}$, $\left|\norm{Ty}_r - \norm{y}_p\right| \le \eps \norm{y}_p$ holds with large constant probability.

    We next prove $(2)$. We first show that for every $x \in \R^d$, $\norm{Ty}_r^r \ge (1 - \eps) \norm{y}_p^r$ holds with probability at least $1 - \exp(-d \log(d)/\eps^{O(1)})$. Recall that we have that we have that $\E|(Ty)_i|^r = \frac{1}{m} \norm{y}_p^r$ for every $i$. Fix $k = 1/\eps^{O(1)}$. Let 
    \[
    s_i = |(Ty)_{(i - 1)k + 1}|^r + |(Ty)_{(i - 1)k + 2}|^r  + \cdots +|(Ty)_{ik}|^r \ \ (1 \le i \le m / k) \;.
   \]    
     We then have $\norm{Ty}_r^r = \sum_i s_i$. Similar to~$(1)$, one can show that for each $i$, with large constant probability 
     \begin{equation}
     \label{eq:2}
        \left|s_i - \frac{k}{m} \norm{y}_p^r \right| \le \eps \frac{k}{m}\norm{y}_p^r \; 
     \end{equation}
     By a Chernoff bound, with probability at least $1 - \exp(-d/\eps^{\Omega(1)})$, at least a $(1 - \eps)$-fraction of the $s_i$ satisfy~\eqref{eq:2}. Conditioned on this event, it holds that
     \[ 
     \norm{Ty}_r^r = \sum_i s_i \ge \frac{m}{k}(1-\eps) \frac{k}{m}\norm{y}_p^r = (1 - \eps) \norm{y}_p^r \;,
     \]
     which is what we need.

    The next is a standard  net-argument. Let $\mathcal{S} = \{Ax: x \in \R^d, \norm{Ax}_p = 1\}$ be the unit $\ell_p$-ball and $\mathcal{N}$ be a $\gamma$-net with $\gamma =  \poly(\eps /d)$ under the $\ell_p$ distance.  It is a standard fact that the size of $\mathcal{N}$ can be $(\poly(d/\eps))^{d}$.  By a union bound, we have that $\norm{TAx}_r \ge (1 - \eps) \norm{Ax}_p = (1 - \eps)$ for all $Ax \in \mathcal{N}$ simultaneously with probability at least $9/10$. From Lemma~\ref{lem:p_stable_upper_bound}, we have that with probability at least $9/10$, $\norm{TAx}_p \le \poly(d) \norm{Ax}_p$ for all $x\in \R^d$. Conditioned on these events, we then have for all $x \in \R^d$,
     \[
     \norm{TAx}_r \le m^{1/r - 1/p} \norm{TAx}_p \le \poly(d/\eps) \norm{Ax}_p \;.
     \]
     Then, for each $y = Ax \in \mathcal{S}$, we choose a sequence of points $y_0,y_1,\dots\in \mathcal{S}$ as follows.
\begin{itemize}
    \item Choose $y_0 \in \mathcal{S}$ such that $\norm{y - y_0}_p \le \gamma$ and let $\alpha_0 = 1$;
    \item After choosing $y_0,y_1,\dots,y_i$, we choose $y_{i+1}$ such that
    \[
    \Norm{ \frac{y - \alpha_0 y_0 - \alpha_1 y_1 - \cdots - \alpha_i y_i}{\alpha_{i+1}} - y_{i+1} }_p \leq \gamma,
    \]
    where $\alpha_{i+1} = \norm{y - \alpha_0 y_0 - \alpha_1 y_1 - \cdots - \alpha_i y_i}_p$. 
\end{itemize}
The choice of $y_{i+1}$ means that
\[
    \alpha_{i+2} = \norm{ y - \alpha_0 y_0 - \alpha_1 y_1 - \cdots - \alpha_i y_i - \alpha_{i+1}y_{i+1} }_p \leq \alpha_{i+1}\gamma.
\]
A simple induction yields that $\alpha_i \leq \gamma^i$. Hence
\[
    y = y_0 + \sum_{i \ge 1} \alpha_i y_i , \ \ \ |\alpha_i| \le \gamma^{i} \;.
\]
Suppose that $y_i = Ax_i$. We have
\[
\norm{TAx}_r \ge \norm{TAx_0}_p - \sum_{i \ge 1} \gamma^i \norm{TAx_i}_p \ge (1 - \eps) - \sum_{i \ge 1} \gamma^i\cdot (\poly(d/\eps))  = 1 - O(\eps).
\]
Rescaling $\eps$, we obtain that $\norm{TAx}_r^r \ge (1 - \eps) \norm{Ax}_p^r$ for all $x\in\R^d$ simultaneously.

This completes the proof of the two guarantees for the original $T$, without rounding the entries. To show that the guarantees continue to hold after rounding the entries, 
        We only need to notice that 
        \begin{align*}
        \abs{ \norm{\tilde{T}Ax}_r - \norm{TAx}_r }  \le \norm{(\tilde{T} - T)Ax}_r 
        & \le m^{\frac{1}{r} - \frac{1}{2}}  \norm{(\tilde{T} - T)Ax}_2 \\
        & \le m^{\frac{1}{r} - \frac{1}{2}} \norm{\tilde{T} - T}_2 \norm{Ax}_2 \\
        & \le \frac{1}{\poly(nd)} \norm{Ax}_p \;. 
        \end{align*}
    \end{proof}

\paragraph{Lewis Weight Sampling.} It is known that sampling $\tilde{O}(d/\eps^2)$ rows with respect to the $\ell_p$ Lewis weights gives an $\ell_p$ subspace embedding with large constant probability when $p\in [1,2]$ \cite{CP15}. 
In the following lemma, we shall show that for  $\ell_p$-regression, sampling $\tilde{O}(d/\eps)$ rows is enough.

\begin{lemma}\label{lem:regression}
Let $A \in \R^{n\times d}$, $b\in\R^n$ and $p\in (1,2)$. Suppose that $S$ is a rescaled sampling matrix according to $w_i([A\ b])$ with oversampling factor $\beta = \Theta(\eps^{-1}\log^2 d \log n \log(1/\delta))$ and $\tilde{x} = \arg \min_{x\in \R^d} \norm{SA x - Sb}_p$. With probability at least $1-\delta$, it holds that
$
   \norm{A\tilde{x}-z}_p \leq (1+\eps) \min_{x\in\R^d} \norm{Ax-z}_p
$
and the number of rows that $S$ samples is
$
O\left(\eps^{-1} d \log^2 d \log n \log(1/\delta) \right)
$.
\end{lemma}

The proof of the lemma closely follows the proof in~\cite{CLS22} and is postponed to Appendix~\ref{sec:regression_proof}.

We are now ready to prove our theorem for distributed $\ell_p$-regression. 

\begin{theorem}\label{thm:ell_p_regression}
    The protocol described in Figure~\ref{alg:ell_p} returns a $(1 \pm \eps)$-approximate solution to the $\ell_p$-regression problem with large constant probability. The communication complexity is $\tilde{O}(sd^2/\eps + sd/\eps^{O(1)})$ and the total runtime of all servers is $O((\sum_i \mathrm{nnz}(A^i))\cdot (d/\eps^{O(1)}) + s \cdot \poly(d/\eps))$.
\end{theorem}

\begin{proof}
    By Lemma~\ref{lem:p_stable_embedding}(1), it holds with high constant probability that
    \[
    \min_{x \in \R^d} \norm{T(Ax - b)}_r \le (1 + \eps) \min_{x \in \R^d} \norm{Ax- b}_p \;.
    \]
    Suppose that $x' \in \R^d$ is a $(1 + \eps)$-approximate solution to $\min_{x \in \R^d} \norm{T(Ax - b)}_r$, i.e.,
    \[
    \norm{T(Ax' - b)}_r \le (1 + \eps) \min_{x \in \R^d} \norm{T(Ax - b)}_r \; .
    \]
    It follows from Lemma~\ref{lem:p_stable_embedding}(2) that 
    \[
    \norm{Ax' - b}_p \le \frac{1}{1 - \eps}\norm{T(Ax' - b)}_r \le (1 + O(\eps)) \min_{x \in \R^d} \norm{Ax- b}_p \;. 
    \]
    Hence, the problem is reduced to obtaining a $(1 + \eps)$-approximate solution to $\min_{x \in \R^d} \norm{T(Ax - b)}_r = \min_{x \in \R^d} \norm{\hat{A} x - \hat{b}}_r$.

    Consider the iteration in Step~3. A standard analysis (see, e.g., Section 2.4 of~\cite{woo14}) yields that in each iteration, with probability at least $1 - 1/\poly(d)$, $\tau$ is a constant approximation to the leverage score of $W^{1/p - 1/2} B$. Taking a union bound, we get that with high constant probability, for all iterations it holds. Conditioned on this event happening, from Lemma~\ref{lem:lewis_algortihm} we get that after $t$ iterations, $w$ is a constant approximation to the $\ell_r$ Lewis weights of $B$ (in each iteration we round $w$; however, notice that if the Lewis weight $w_i$ is not $0$, it should be larger than $1/\poly(nd)$ as the non-zero entries of the matrix $B$ are at least $1/\poly(nd)$\footnote{It is easy to see that the $\ell_r$ sensitivities, defined in Proposition~\ref{prop:sensitivity_bound} in Section~\ref{sec:runtime_appendix}, are at least $\Omega(1/\poly(nd))$ in our setting if the corresponding rows are nonzero as if we take $x = a_i$, we can get that the $\ell_i^{(r)}(A) \ge \norm{a_i}^{2p}/\norm{Aa_i}_p^p$ where the denominator is at most $\poly(nd)$ as each entry of $A$ is in $\poly(nd)$. From Lemma~2.5 in~\cite{MMW+22}, we know that the $\ell_r$ Lewis weights are larger than the $\ell_r$ sensitivities when $r < 2$.
    }  
    , and hence the rounding will not affect the approximation ratio guarantee in each iteration). From Lemma~\ref{lem:regression}, the solution to  $\min_{x \in \R^d} \norm{A'x - b'}_r$ is a $(1 + \eps)$-approximate solution to $\min_{x \in \R^d} \norm{T(Ax - b)}_r$, and is thus a $(1 \pm O(\eps))$-approximate solution to the original problem $\min_{x \in \R^d} \norm{Ax- b}_p$.  

    We next analyze the communication complexity of the protocol. For Step 3(a), $S_t \tilde{W}^{1/2 - 1/p} B_i$ is a $d \log(d) \times (d +1)$ matrix and the entries of $S_t \tilde{W}^{1/2 - 1/p} B_i$ are in $\poly(nd)$-precision as the entries of $S_t, \tilde{W}^{1/2 - 1/p}$, and $B_i$ are both in $\poly(nd)$-precision. Hence, the total communication of all servers is $\tilde{O}(sd^2)$. For Step 3(b), $\tilde{R}$ is a $(d + 1) \times (d + 1)$ matrix and hence the total communication cost is $\tilde{O}(sd^2)$. For 3(c), $B^i \tilde{R} G$ is a $d / \eps^{O(1)} \times O(\log d)$ matrix, and hence similarly we get that the total communication cost is $O(sd/\eps^{O(1)})$. For 3(e), since $w$ is a $d/\eps^{O(1)}$ vector, the total communication cost of this step is $O(sd/\eps^{O(1)})$. In Step 5, since the sum of Lewis weights is $O(d)$, with high constant probability the server samples at most $\tilde{O}(d/\eps)$ rows, and hence the communication cost of this step is $O(sd^2/\eps)$. Putting everything together, we get that the total communication cost is
    \[
    \tilde{O}\left(\log \log(d/\eps) \cdot(sd^2 + sd/\eps^{O(1)}) + sd^2/\eps\right) = \tilde{O}(sd^2/\eps + sd/\eps^{O(1)})\; . 
    \]
    
     We now consider the runtime of the protocol. To compute $T A^i$, notice that $T$ has $d/\eps^{O(1)}$ rows, which means it takes $O(\nnz(A^i)\cdot (d/\eps^{O(1)}))$ times to compute $TA^i$. Hence Step 2 takes time $O((\sum_i \nnz(A^i))\cdot (d/\eps^{O(1)}))$. For the remaining steps, one can verify that each step takes $\poly(d/\eps)$ time on a single server or on the coordinator. The total runtime is therefore $O(\sum_i \mathrm{nnz}(A^i)\cdot (d/\eps^{O(1)}) + s \cdot \poly(d/\eps))$.
\end{proof}

We remark that when all leverage scores of $[A\ b]$ are $\poly(\eps)/d^{4/p}$, the servers can first uniformly sample $O(\poly(\eps)/d\cdot n)$ rows of $A$ using the public random bits, rescale the sampled rows and obtain an $A'$. The servers can then run the protocol on $A'$. This modified protocol will still produce a $(1+\eps)$-approximate solution to the $\ell_p$-regression problem and has the same communication complexity because uniform sampling does not require communication. The runtime is now reduced to $O(\sum_i \mathrm{nnz}(A^i) + s \cdot \poly(d/\eps))$, which is optimal in terms of $\nnz(A^i)$. The details, including the formal statement, can be found in Appendix~\ref{sec:runtime_appendix}.

\section*{Acknowledgements}
Y. Li is supported in part by Singapore Ministry of Education (AcRF) Tier 1 grant RG75/21 and Tier 2 grant MOE-T2EP20122-0001.  H. Lin and D. Woodruff would like to thank support from the
National Institute of Health (NIH) grant 5R01 HG 10798-2 and the Office of Naval
Research (ONR) grant N00014-18-1-2562. 

\bibliography{reference} 
\bibliographystyle{alpha}

\appendix

\section{Proof of Lemma~\ref{lem:set}}
\label{sec:set}

We need the following theorem on the singularity probability of random sign matrices.

\begin{theorem}[\cite{TIK20}]\label{thm:singularity_main}
Let $M_n \in \mathbb{R}^{n \times n}$ be a random matrix whose entries are i.i.d.\ Rademacher random variables. It holds that 
\[
\Pr\left[M_n \text{ is singular}\right] \le (\frac{1}{2} + o_n(1))^{n}\;.
\]
\end{theorem}

The proof of the following lemma follows directly from the proof in~\cite{VWW20} with only minor modifications.
\begin{lemma}\label{lem:prob_method}
For sufficiently large $d$, there exists a set of matrices $\mathcal{T} \subseteq {\{-1,1\}}^{d \times d}$ with $|\mathcal{T}| = \Omega(2^{0.49d^2})$ such that
\begin{enumerate}
\item For any $T \in \mathcal{T}$, $\rank(T) = d$;
\item For any $S, T \in \mathcal{T}$ such that $S \neq T$, $\spa([S_{d-1}~T_{d - 1}]) = \mathbb{R}^d$, where $S_{d - 1}$ denotes the first $d - 1$ column of $S$.
\end{enumerate}
\end{lemma}
\begin{proof}
We use the probabilistic method to prove the existence. Let $t = 2 - \eps$, where $\eps$ is a sufficiently small constant.
We use $\bado \subset \mathbb{R}^{d \times (d - 1)}$ to denote the set
$$
\bado = \{B \in \mathbb{R}^{d \times (d - 1)} \mid \Pr[X \in \mathrm{span}(B)] \ge c \cdot t^{-d}\text{ or } \rank(B) < d - 1\},
$$
where $X \in \mathbb{R}^d$ is a vector whose entries are i.i.d.\ Rademacher variables and $c$ is an absolute constant.

Consider a random matrix $A \in \mathbb{R}^{d \times (d - 1)}$ with i.i.d.\ Rademacher entries. Then
\begin{equation}\label{pr:bad}
\Pr[A \in \bado] \le \frac{1}{c},
\end{equation}
since otherwise, if we use $X \in \mathbb{R}^d$ to denote a vector with i.i.d.\ Rademacher coordinates, we have
\begin{align*}
&\Pr[\rank([A~X]) < d] \\
\ge& \Pr[\rank([A~X]) < d \mid  A \in \bado] \cdot \Pr[ A \in \bado]  \\
> &t^{-d},
\end{align*}
which violates Theorem \ref{thm:singularity_main}.

For any fixed $A \in \mathbb{R}^{d \times (d - 1)} \setminus \bado$, 
consider a random matrix $B \in \mathbb{R}^{d \times (d - 1)}$ whose entries are i.i.d.\ Rademacher variables, 
\begin{equation}\label{equ:span}
\Pr[\spa([A~B]) = \mathbb{R}^d] \ge 1 -  \Pr\left[\bigcap_{i = 1}^{d - 1} B_i \in \spa(A)\right] \ge 1 - c^d t^{-d(d -1 )},
\end{equation}
which follows from the definition of $\bado$ and the independence of the columns of $B$.

Now we construct a multiset $\mathcal{S}$ of $c^{-d}t^{d(d - 1) / 2}$ matrices, chosen uniformly with replacement from $\{-1, 1\}^{d\times d}$.
By (\ref{pr:bad}) and linearity of expectation, we have 
$$\E[|\mathcal{S} \cap S_\bado|] \le c^{-d}t^{d(d - 1) / 2} \cdot \frac{1}{c}, $$
where $S_{\bado}$ denotes the set of the matrices $M$ such that the first $d - 1$ columns of $M$ is in $\bado$. Let $\mathcal{E}_1$ denote the event that 
$$|\mathcal{S} \cap S_\bado| \le 4\E[|\mathcal{S} \cap S_\bado|] \le 4 c^{-(d+1)}t^{d(d - 1) / 2},$$ 
which holds with probability at least $3 / 4$ by Markov's inequality. 

Let $S_{\mathsf{rank}}$ denote the set of $d \times  d$ matrices that are not of full rank. By (\ref{pr:bad}) and linearity of expectation, we have 
\[
\E[|\mathcal{S} \cap S_{\mathsf{rank}}|] \le c^{-d}t^{d(d - 1) / 2} \cdot t^{-d}, 
\]
Let $\mathcal{E}_2$ denote the event that 
\[
|\mathcal{S} \cap S_{\mathsf{rank}}| \le 4\E[|\mathcal{S} \cap S_\mathsf{rank}|] \le4 c^{-d}t^{d(d - 1) / 2} \cdot t^{-d}, 
\]
which holds with probability at least $3 / 4$ by Markov's inequality. 

Let $\mathcal{E}_3$ denote the event that 
$$
\forall S \in \mathcal{S} \setminus S_\bado, \forall T \in \mathcal{S} \setminus \{S\}, \spa([S_{d - 1}~T_{d - 1}]) = \mathbb{R}^d.
$$
Using a union bound and (\ref{equ:span}), 
$$
\Pr(\mathcal{E}_3) \geq 1 - |\mathcal{S}|^2 c^d t^{-d(d - 1)} = 1 - o_d(1).
$$

Thus by a union bound, the probability that all $\mathcal{E}_1$, $\mathcal{E}_2$ and $\mathcal{E}_3$  hold is strictly larger than zero, which implies there exists a set $\mathcal{S}$ such that $\mathcal{E}_1$ $\mathcal{E}_2$, and $\mathcal{E}_3$ hold simultaneously. 
Now we consider $\mathcal{T} = \mathcal{S} \setminus (S_\bado \cup S_{\mathsf{rank}})$. 
Since $\mathcal{E}_1$ and $\mathcal{E}_2$ hold, we have $|\mathcal{T}| \ge \Omega(c^{-2d}t^{d(d - 1) / 2}) = \Omega(2^{0.49d^2})$, provided that $d$ is sufficiently large and $\eps$ is sufficiently small.
The event $\mathcal{E}_3$ implies that all elements in $\mathcal{T}$ are distinct, and furthermore, it holds for any $S, T \in \mathcal{T}$ with $S \neq T$ that 
$\spa([S_{d-1}~T_{d - 1}]) = \mathbb{R}^d$.
\end{proof}

Suppose that $\mathcal{T}$ satisfies the conditions in Lemma~\ref{lem:prob_method}. 
For each $T \in \mathcal{T}$, we add $T^T$ into $\mathcal{H}$. 
Now suppose there exist $S, T \in \mathcal{H}$ such that $S\neq T$ and $S^{-1} e_d = T^{-1} e_d$, which means there exists $x \in \mathbb{R}^d$ such that $Sx = e_d$ and $Tx = e_d$.
This implies that $x^T (S^T)_{d-1} = x^T (T^T)_{d-1} = 0$. The construction of $\mathcal{T}$ guarantees that $\spa([(S^T)_{d-1} ~ (T^T)_{d-1}]) = \mathbb{R}^d$ and it must thus hold that $x = 0$, which would result in $Sx = Tx = 0\neq e_d$.
%
Therefore, for any $S, T \in \mathcal{H}$ with $S\neq T$, $S^{-1} e_d \neq T^{-1} e_d$.


\section{Proof of Lemma~\ref{lem:regression}}\label{sec:regression_proof}

The proof of the lemma closely follows that in~\cite{CLS22}. The proof is a bootstrapping argument based on the following two lemmas. For simplicity of notation, we define $R(A,b) = \min_x \norm{Ax-b}_p$.

\begin{lemma}[{\cite[Theorem 3.18]{MMW+22}}] \label{lem:gamma-correctness}
There exists an absolute constant $c\in (0,1]$ such that the following holds for all $A \in \R^{n\times d}$, $b \in \R^n$ and $\gamma \in (0,1)$. Let $x^* = \arg\min_{x\in \R^d} \norm{Ax - b}_p$. Whenever $x\in\R^d$ satisfies $\norm{Ax-b}_p \leq (1+c\gamma)R(A,b)$, we have that $\norm{Ax^*-Ax}_p \leq \sqrt{\gamma} R(A,b)$.
\end{lemma}

\begin{lemma}\label{lem:gamma-assumption-eps-approx}
Let $A \in \R^{n \times d}$, $b\in \R^n$ and $0< \gamma <1$. Let $S$ be the rescaled sampling matrix with respect to $\{ p_{i} \}_{(i)}$ such that $p_{i}=\min \{\beta w_{i}([A\ b]), 1 \}$ and $\beta = \Theta(\frac{\gamma}{\eps^2} \log^2 d \log n \log \frac{1}{\delta})$. Suppose that $\tilde{x} = \arg\min_{x\in \R^d} \norm{SA x - Sb}_p$ and $\norm{A \tilde{x} - Ax^*}_p \leq \sqrt{\gamma} R(A,b)$. It holds that
\[
    \norm{A\tilde{x} - b}_p \leq (1 + C\eps) R(A,b)
\]
with probability at least $0.99 - \delta$, where $C$ is an absolute constant.
\end{lemma}

Assuming these two lemmas, the proof of Lemma~\ref{lem:regression} is nearly identical to that in the proof of~\cite[Theorem 5.6]{CLS22} and is thus omitted. The proof is simpler because we do not need an argument to first show that sketching by $S$ gives a $(1+O(\sqrt{\eps}))$-approximate solution, which follows immediately from the fact that $S$ is a $(1+\sqrt{\eps})$-subspace-embedding with large constant probability.

In the remainder of this section, we discuss the proof of Lemma~\ref{lem:gamma-assumption-eps-approx}. The proof is similar to that of~\cite[Lemma 5.4]{CLS22}, which converts the bound on the target dimension obtained from an iterative argument in~\cite{MMW+22} to a moment bound using the framework in~\cite{CP15}.

The difference is that here we can choose the weights to be the Lewis weights of $[A b]$, while in \cite[Lemma 5.4]{CLS22}, it considers $\min_x\norm{Ax-z}_p$ with $\norm{z}_p \leq R(A,b)$ and it samples the rows according to the Lewis weights of $A$. Specifically, let $R = R(A,b)$, $x' = x -x^*$ and $b' = b - Ax^*$. We have, as in the proof of \cite[Lemma 5.4]{CLS22}, that
\begin{align*}
    \norm{A\tilde{x}-b}_p^p - \norm{Ax^*-b}_p^p &= \norm{A\tilde{x}-b}_p^p - \norm{SA\tilde{x} - Sb}_p^p + \norm{SA\tilde{x} - Sb}_p^p - \norm{SAx^* - Sz}_p^p\\
    &\qquad \qquad + \norm{SAx^* - Sb}_p^p - \norm{Ax^* - b}_p^p\\
    &\leq \norm{A\tilde{x}-b}_p^p - \norm{SA\tilde{x} - Sb}_p^p + \norm{SAx^* - Sb}_p^p - \norm{Ax^* - b}_p^p\\
    &\leq \Norm{Ax' - b'}_p^p - \Norm{SAx' - Sb'}_p^p + \Norm{Sb'}_p^p - \Norm{b'}_p^p \\
    &= \Norm{Ax' - b'}_p^p - \Norm{Ax' - \bar{b}'}_p^p - \Norm{b' - \bar{b}'}_p^p \\
    &\qquad - \left( \Norm{SAx' - Sb'}_p^p - \Norm{SAx' - S\bar{b}'}_p^p - \Norm{Sz' - S\bar{b}'}_p^p \right) \\
    &\qquad - \left( \Norm{SA x'-S \bar{b}'}_p^p - \Norm{Ax' - \bar{b}'}_p^p + \Norm{\bar{b}'}_p^p - \Norm{S\bar{b}'}_p^p \right) \\
    &=: E_1 - E_2 - E_3,
\end{align*}
where $\bar{b}$ is the vector obtained from $b$ by removing all coordinates $b_i$ such that $|b_i|\geq \frac{w_i}{\eps}R$. Note that $\norm{\bar{b}'}_p\leq \norm{b'}_p = R$ and $\norm{Ax'}_p\leq\sqrt{\gamma}R$. The first term can be controlled using \cite[Lemma 3.5]{MMW+22}, except that the sampling probabilities are Lewis weights of $[A\ b]$ instead of $A$, but the proof still goes through because it also holds that $|(Ax)_i|^p\leq \norm{Ax}_p s_i^p([A\ b])$, where $s_i([A\ b])$ is the $\ell_p$-sensitivity of $[A\ b]$. The second term can be controlled by \cite[Lemma 3.6]{MMW+22}, yielding that $|E_2|\leq \eps R^p$ with probability at least $0.99$. The last term can be controlled as in~\cite[Lemma 5.3]{CLS22}, where the Lewis weights of $[A\ b]$ do not affect the proof.

\section{Faster Runtime for Distributed $\ell_p$-Regression} \label{sec:runtime_appendix}

We need the following auxiliary results.

\begin{proposition}\label{prop:sensitivity_bound}
Suppose that $1\leq p < 2$ and $A\in \R^{n\times d}$. The $\ell_p$-sensitivity scores of $A$ are defined as
    \[
    \ell_i^{(p)}(A) = \sup_{x: Ax\neq 0} \frac{\abs{\langle a_i, x\rangle}^p}{\norm{Ax}_p^p},
    \]
    where $a_i$ is the $i$-th row of $A$. It holds that $\ell_i^{(p)}(A)\leq (\tau_i(A))^{p/2}$ for all $i$.
\end{proposition}
\begin{proof}
    Suppose that $A$ has full column rank, otherwise we can find an invertible matrix $T$ such that $AT = [A'\ 0]$, where $A'$ has full column rank, and consider $\ell_i^{(p)}(A')$ and $\tau_i(A')$ instead. It is not difficult to verify that $\ell_i^{(p)}(A') = \ell_i^{(p)}(A)$ and $\tau_i(A') = \tau_i(A)$.

    Write $A = UR$, where $U\in \R^{n\times d}$ has orthonormal columns and $R\in \R^{d\times d}$ is invertible. Then 
    \[
        \ell_i^{(p)}(A) = \sup_{y\neq 0} \frac{\abs{\inner{U_i, y}}^p}{\norm{Uy}_p^p} 
        \leq \sup_{y\neq 0} \frac{\norm{U_i}_2^p \norm{y}_2^p}{\norm{Uy}_2^p} 
        = \norm{U_i}_2^p 
        = (\tau_i(A))^{p/2},
    \]
    as advertised.
\end{proof}

\begin{lemma}[{\cite[Lemma 5.5]{LLW23}}]
\label{lem:sensitivity}
Let $A \in \R^{n \times d}$ and $1 \le p < \infty$. The matrix $A'$ is a submatrix of $A$ such that the rescaled $i$-th row $p_i^{-1/p}a_i$ is included in $A'$ with probability $p_i \ge \min(\beta s_i(A), 1)$. Then, there is a constant $c$ such that when $\beta \ge c \eps^{-2} d \log(1/\eps)$,  the matrix $A'$ is a $(1 \pm \eps)$-subspace embedding of $A$ with probability at least $9/10$.
\end{lemma}

As an immediate corollary of the auxiliary results above, we have that when $A\in\R^{n\times d}$ has uniformly small leverage scores, uniformly sampling its rows can give an $\ell_p$-subspace-embedding (after rescaling).
\begin{corollary}
Suppose that $1\leq p<2$ and the matrix $A\in\R^{n\times d}$ satisfies that $\tau_i(A) \leq (c\eps^2\gamma/(d\log(1/\eps)))^{2/p}$ for all $i$, where $\gamma \leq \eps^2/(Cd\log(1/\eps))$. Let $A'$ be a matrix formed from $A$ by retaining each row with probability $\gamma$ independently and then rescaling by $1/\gamma^{1/p}$. It holds with large constant probability that
\[
(1-\eps)\norm{Ax}_p^p \leq \norm{A'x}_p^p \leq (1+\eps)\norm{Ax}_p^p
\]
for all $x\in \R^d$ simultaneously, and that $A'$ has $O(\gamma n)$ rows.
\end{corollary}
\begin{proof}

    By Proposition~\ref{prop:sensitivity_bound}, $\ell_i^{(p)}(A)\leq (\tau_i(A))^{p/2} = c\eps^2\gamma/(d\log(1/\eps))$, so the sampling probability 
    \[
        \gamma\geq \frac{Cd\log(1/\eps)}{\eps^2}\cdot \ell_i^{(p)}(A)
    \]
    satisfies the condition in Lemma~\ref{lem:sensitivity}. The conclusion follows immediately.
\end{proof}

Hence, if $A$ has uniformly small leverage scores, all sites can agree on the $O(\gamma n)$ uniformly sampled rows using the public random bits and run the protocol in Algorithm~\ref{alg:ell_p} on the induced $A'$. By Markov's inequality, $\nnz(A') = O(\gamma\nnz(A))$ with large constant probability and we finally conclude with the following theorem.

\begin{theorem}
    Suppose that $A\in\R^{n\times d}$ and $b\in \R^d$ satisfies that the leverage scores of $[A\ b]$ are all bounded by $\poly(\eps)/d^{4/p}$.
    There is a protocol which outputs a $(1 \pm \eps)$-approximate solution to the $\ell_p$-regression problem with large constant probability, using $\tilde{O}(sd^2/\eps + sd/\eps^{O(1)})$ bits of communication and running in total time (over all servers) $O(\sum_i \mathrm{nnz}(A^i) + s \cdot \poly(d/\eps))$.
\end{theorem}

\end{document}